\documentclass[11pt]{article}

\usepackage[text={17.2cm,23cm},centering]{geometry}
\usepackage[format=hang,indention=-3em,textfont={small,it},labelfont={small,bf},width=0.92\textwidth]{caption}
\usepackage{graphicx}
\usepackage{subfig}
\usepackage{amsmath,amssymb,amsthm,mathabx,MnSymbol,upgreek}
\usepackage{citesort}
\usepackage{enumerate}
\usepackage{authblk}
\usepackage{appendix}
\usepackage{stackrel}

\usepackage[square,numbers]{natbib}

\usepackage[algoruled,vlined,boxed,longend,english,lined,boxed,commentsnumbered]{algorithm2e} 	
\usepackage{algorithmic}

\newtheoremstyle{mythm}{2ex}{2ex}{\itshape}{}{\color{darkgray}\normalfont\normalsize\bfseries\sffamily\slshape}{}{.5em}{\thmname{#1}\thmnumber{ #2}\thmnote{ (#3)}}
\newtheorem{prop}{Proposition}
\newtheorem{remark}{Remark}

\newcommand{\periodafter}[1]{#1.}
\usepackage{titlesec}
\usepackage{titletoc}
\titleformat{\section}{\color{RoyalBlue3}\normalfont\large\sffamily\bfseries}{\thesection}{1em}{}
\titleformat{\subsection}[runin]{\color{RoyalBlue4}\normalfont\normalsize\sffamily\bfseries}{\thesubsection}{1em}{\periodafter}
\titleformat{\subsubsection}[runin]{\color{RoyalBlue4}\normalfont\normalsize\sffamily\bfseries\itshape}{\thesubsubsection}{1em}{\periodafter}
\titleformat{\paragraph}[runin]{\color{RoyalBlue4}\normalfont\normalsize\slshape\bfseries\sffamily}{}{}{\periodafter}
\titlespacing{\section}{0pt}{3.5ex plus 1ex minus .2ex}{2.3ex plus .2ex}
\titlespacing{\subsection}{0pt}{3.5ex plus 1ex minus .2ex}{2.3ex plus .2ex}
\titlespacing{\subsubsection}{0pt}{3.25ex plus 1ex minus .2ex}{1ex plus .2ex}
\titlespacing{\paragraph}{0pt}{3ex plus 1ex minus .2ex}{1em}
\titlespacing*{\subparagraph}{1.3em}{2ex plus 1ex minus .2ex}{1em}
\usepackage[usenames,dvipsnames,x11names,table]{xcolor}

\graphicspath{{figures/}}

\makeatletter
\newcommand{\proofstep}[1]{%
  \par
  \addvspace{\smallskipamount}
  \textit{#1\@addpunct{.}}\enspace\ignorespaces}
\makeatother

\usepackage{autonum}

\usepackage[notref,notcite,color,final]{showkeys}      
\definecolor{refkey}{named}{blue}
\definecolor{labelkey}{named}{blue}

\newcommand{\nes}{\hspace*{-0.6pt}}
\newcommand{\exs}{\hspace*{0.6pt}}
\newcommand{\alphad}{\dot{\alpha}}

\newcommand{\bfA} {\boldsymbol{A}}

\newcommand{\bfal}{\boldsymbol{\alpha}}
\newcommand{\bfald}{\dot{\bfal}}

\newcommand{\bfe} {\boldsymbol{e}}

\newcommand{\bfeps}{\boldsymbol{\varepsilon}}
\newcommand{\bfepsd}{\dot{\bfeps}}

\newcommand{\bfI} {\boldsymbol{I}}

\newcommand{\bfs} {\boldsymbol{s}}

\newcommand{\bfsig}{\boldsymbol{\sigma}}
\newcommand{\bfsigup}{\boldsymbol{\upsigma}}
\newcommand{\bfepsup}{\boldsymbol{\upepsilon}}

\newcommand{\bfsige}{\bfsig^{\mathrm{e}}}

\newcommand{\bfsigv}{\bfsig^{\mathrm{v}}}

\newcommand{\bfu} {\boldsymbol{u}}

\newcommand{\ii}{\text{i}}

\newcommand{\bfx} {\boldsymbol{x}}
\newcommand{\bfxi} {\boldsymbol{\xi}}

\newcommand{\bfze}{\mathbf{0}}

\newcommand{\CA}{\CS^{\star}_{\!A}}
\newcommand{\Ca}{\CS_{\alpha}}
\newcommand{\Caps}{\CS_{\alpha}^{\tiny\text{ps}}}
\newcommand{\Daps}{\DS_{\alpha}^{\tiny\text{ps}}}

\newcommand{\CaS}{\Ca\Ssup}

\newcommand{\Pcal}{\mathcal{P}}
\newcommand{\Qcal}{\mathcal{Q}}
\newcommand{\Rcal}{\mathcal{R}}

\newcommand{\Ce}{\CS_{\eps}}
\newcommand{\Ceps}{\CS_{\eps}^{\tiny\text{ps}}}
\newcommand{\Deps}{\DS_{\eps}^{\tiny\text{ps}}}
\newcommand{\pssup}{^{\tiny\text{ps}}}

\newcommand{\SSs}{\CSS_{\nes\sigma}}
\newcommand{\STs}{\CST_{\nes\sigma}}
\newcommand{\SSA}{\CSS_{\!A}}
\newcommand{\STA}{\CST_{\!A}}
\newcommand{\CeS}{\Ce\Ssup}
\newcommand{\CH}{\widehat{\CS}}
\newcommand{\CHT}{\widehat{\CS}{}\Tsup}
\newcommand{\Cm}{\CS_{\mathrm{m}}}
\newcommand{\Cms}{\CS^{\star}_{\mathrm{m}}}

\newcommand{\CmT}{\CS_{\mathrm{m}}}
\newcommand{\CmTps}{\CS_{\mathrm{m}}^{\tiny\text{ps}}}
\newcommand{\DmTps}{\DS_{\mathrm{m}}^{\tiny\text{ps}}}
\newcommand{\Cs}{\CS^{\star}_{\sigma}}

\newcommand{\CS} {\text{\boldmath $\mathcal{C}$}}
\newcommand{\CSS}{\text{\boldmath $\mathcal{S}$}}
\newcommand{\CST}{\text{\boldmath $\mathcal{T}$\!\nes}}

\newcommand{\DA}{\DS^{\star}_{\!A}}
\newcommand{\Da}{\DS_{\nes\alpha}}

\newcommand{\DaS}{\Da\Ssup}

\newcommand{\De}{\DS_{\eps}}

\newcommand{\bfLa}{\boldsymbol{\Lambda}}

\newcommand{\DeS}{\De\Ssup}
\newcommand{\del}[1][]{\partial_{#1}}
\newcommand{\demi} {\tfrac{1}{2}}
\newcommand{\der}[2]{\dfrac{\text{d}#1}{\text{d}#2}}
\newcommand{\dip} {\! :\!}
\newcommand{\Dm}{\DS_{\mathrm{m}}}
\newcommand{\DmT}{\DS_{\mathrm{m}}}
\newcommand{\Dms}{\DS^{\star}_{\mathrm{m}}}

\newcommand{\dotp}{\raisebox{2pt}{\hspace*{1pt}\scalebox{0.5}{$\bullet$}}\hspace*{1pt}}

\newcommand{\DS}{\mbox{\boldmath $\mathcal{D}$}}
\newcommand{\Ds}{\DS^{\star}_{\sigma}}

\newcommand{\dd}{\,\text{d}}
\newcommand{\dtau}{\,\text{d}\tau}
\newcommand{\dom}{\;\text{d}\omega}
\newcommand{\dth}{\;\text{d}\vartheta}
\newcommand{\ds}{\,\text{d}s}
\newcommand{\dt}{\,\text{d}t}

\newcommand{\epsd}{\dot{\eps}}
\newcommand{\eps}{\varepsilon}

\newcommand{\Fcal}{\mathcal{F}}
\newcommand{\FS}{\ensuremath{\mbox{\boldmath $\mathcal{F}$}}}

\newcommand{\inv}[1]{\dfrac{1}{#1}}

\newcommand{\IS}{\mbox{\boldmath $\mathcal{I}$}}
\newcommand{\Isub}{_{\text{\tiny I}}}
\newcommand{\JS}{\boldsymbol{\mathcal{J}}}

\newcommand{\KS}{\mbox{\boldmath $\mathcal{K}$}}

\newcommand{\Lcal}{\mathcal{L}}

\newcommand{\Lpar}{\Big(}
\newcommand{\lpar}{\big(}

\newcommand{\oo}{\omega}

\newcommand{\sign}{\text{sgn}}
\newcommand{\bdot}{\boldsymbol{\cdot}}
\newcommand{\QS}{\mbox{\boldmath $\mathcal{Q}$}}

\newcommand{\Rbb} {\mathbb{R}}

\newcommand{\Rpar}{\Big)}
\newcommand{\rpar}{\big)}

\newcommand{\kaC}{\kappa_{\textit{\nes\tiny C}}}
\newcommand{\kaS}{\kappa^{\sigma}}
\newcommand{\kaD}{\kappa_{\textit{\nes\tiny D}}}

\newcommand{\muC}{\mu_{\textit{\nes\tiny C}}}
\newcommand{\muD}{\mu_{\textit{\nes\tiny D}}}

\newcommand{\hkaS}{\hat{\kappa}{}^{\sigma}}

\newcommand{\hkaD}{\!\!\!\hat{\,\,\,\kaD}}
\newcommand{\hmuD}{\!\!\hat{\,\,\muD}}
\newcommand{\hmuC}{\!\!\hat{\,\,\muC}}

\newcommand{\hkaC}{\!\!\!\hat{\,\,\,\kaC}}
\newcommand{\hkaCS}{\hkaC^{\!\!\!\sigma}}
\newcommand{\hmuCS}{\hmuC^{\!\!\!\sigma}}
\newcommand{\kaCS}{\kaC^{\sigma}}
\newcommand{\muCS}{\muC^{\sigma}}

\newcommand{\Csub}{_{\textit{\nes\tiny C}}}

\newcommand{\bfB} {\boldsymbol{B}}
\newcommand{\bfC} {\boldsymbol{C}}
\newcommand{\bfD} {\boldsymbol{D}}
\newcommand{\bfM} {\boldsymbol{M}}
\newcommand{\bfE} {\boldsymbol{E}}

\newcommand{\bfJ} {\boldsymbol{J}}
\newcommand{\bfS} {\boldsymbol{S}}
\newcommand{\bfU} {\boldsymbol{U}}

\newcommand{\bbA} {\textbf{A}}
\newcommand{\bbB} {\textbf{B}}

\newcommand{\bfP} {\boldsymbol{P}}
\newcommand{\bfQ} {\boldsymbol{Q}}
\newcommand{\bfR} {\boldsymbol{R}}
\newcommand{\bbC} {\text{\bf C}}
\newcommand{\bbP} {\text{\bf P}}
\newcommand{\bbR} {\text{\bf R}}
\newcommand{\matp} {\mathfrak{p}}
\newcommand{\matq} {\mathfrak{q}}
\newcommand{\matr} {\mathfrak{r}}
\newcommand{\mats} {\mathfrak{s}}

\newcommand{\hh} {\hspace*{0.5pt}}

\newcommand{\shdeq}{\hspace*{-0.1em}:=\hspace*{-0.1em}}
\newcommand{\sheq}{\hspace*{-0.1em}=\hspace*{-0.1em}}

\newcommand{\shm}{\hspace*{-0.1em}-\hspace*{-0.1em}}
\newcommand{\shneq}{\hspace*{-0.1em}\not=\hspace*{-0.1em}}
\newcommand{\shp}{\hspace*{-0.1em}+\hspace*{-0.1em}}

\newcommand{\sip} {\! \cdot\!}

\newcommand{\Ssup}{^{\text{\tiny S}}}

\newcommand{\tdemi} {\tfrac{1}{2}\exs}
\newcommand{\tens}{\hspace*{-1pt}\otimes\hspace*{-1pt}}

\newcommand{\Tsup}{^{\text{\tiny T}}}
\newcommand{\psupp}{^{\sigma}}
\newcommand{\psup}{^{\text{\tiny (12)}}}
\newcommand{\rsup}{^{\text{\tiny (3)}}}

\newcommand{\psupe}{^{\text{e{\tiny (12)}}}}
\newcommand{\psupv}{^{\text{v{\tiny (12)}}}}

\newcommand{\Dcal}{\mathcal{D}}

\newcommand{\Hcal}{\mathcal{H}}
\newcommand{\Ncal}{\mathcal{N}}
\newcommand{\Ocal}{\mathcal{O}}

\newcommand{\Scal}{\mathcal{S}}

\renewcommand{\del}[1]{\partial_{#1}}

\newcommand{\hBS}{\hat{\bfB}{}^{\sigma\!}}
\newcommand{\hCS}{\hat{\bfC}{}^{\sigma\!}}
\newcommand{\hCE}{\hat{\bfE}{}^{\sigma\!}}
\newcommand{\hCJ}{\hat{\bfJ}{}^{\sigma\!}}

\newcommand{\bfCJ}{\bfJ^{\sigma\!}}
\newcommand{\tauT}{\tilde{\tau}}
\newcommand{\thetaT}{\tilde{\theta}}
\newcommand{\DOT}{\protect\scalebox{0.45}{$\bullet$}}

\title{\color{RoyalBlue3}\normalfont\Large\slshape\bfseries\sffamily
On the constitutive behavior of linear viscoelastic solids \\ under the plane stress condition}

\author[1]{Bojan B. Guzina\thanks{Corresponding author (guzin001@umn.edu)}} 
\author[2]{Marc Bonnet} 
\affil[1]{\small{Dept. of Civil, Environmental, and Geo-Engineering, University of Minnesota, Twin Cities}\vspace*{1mm}} 
\affil[2]{{\small POEMS (CNRS-INRIA-ENSTA), Dept. of Applied Mathematics, ENSTA Paris, France}}

\date{\today}

\begin{document}
\maketitle

\begin{abstract}

\noindent Motivated by the recent experimental and analytical developments enabling high-fidelity material characterization of (heterogeneous) sheet-like solid specimens, we seek to elucidate the constitutive behavior of linear viscoelastic solids under the plane stress condition. More specifically, our goal is to expose the relationship between the plane-stress viscoelastic constitutive parameters and their (native) ``bulk'' counterparts. To facilitate the sought reduction of the three-dimensional (3D) constitutive behavior, we deploy the concept of projection operators and focus on the frequency-domain behavior by resorting to the Fourier transform and the mathematical framework of tempered distributions, which extends the Fourier analysis to functions (common in linear viscoelasticity) for which Fourier integrals are not convergent. In the analysis, our primary focus is the on class of linear viscoelastic solids whose 3D rheological behavior is described by a set of constant-coefficient ordinary differential equations, each corresponding to a generic arrangement of ``springs'' and ``dashpots''. On reducing the general formulation to the isotropic case, we proceed with an in-depth investigation of viscoelastic solids whose bulk and shear modulus each derive from a suite of classical ``spring and dashpot'' configurations. To enable faithful reconstruction of the 3D constitutive parameters of natural and engineered solids via (i) thin-sheet testing and (ii) applications of the error-in-constitutive-relation approach to the inversion of (kinematic) sensory data, we also examine the reduction of thermodynamic potentials describing linear viscoelasticity under the plane stress condition. The analysis is complemented by a set of analytical and numerical examples, illustrating the effect on the plane stress condition on the behavior of isotropic and anisotropic viscoelastic solids.

\end{abstract}

\section{Introduction}

\noindent In medicine, the use of diagnostic ultrasound and magnetic resonance imaging  has paved the way toward quantitative remote sensing of \emph{interior deformation} \mbox{\cite{sigrist2017ultrasound, mariappan2010magnetic}}, which can then be used as sensory data for the constitutive characterization of soft tissues in service of pathology. The class of inverse solutions that operate on this premise are commonly known as \emph{elastography} \cite{parker2005unified,sigrist2017ultrasound}. In the context of dissipative constitutive behavior, there are multiple avenues to elastography including direct algebraic inversion \cite{oliphant2001complex,sinkus2005imaging}, adjoint state techniques \cite{tan2016gradient} and error-in-constitutive-relation (ECR) approach (see~\cite{diaz2015modified,jmps2024} and references therein), that commonly seek to reconstruct the maps of (heterogeneous) viscoelastic properties that are ``most compatible" with the full-field interior kinematic data. 

In experimental mechanics, recent advances in digital image correlation (DIC) and scanning laser Doppler vibrometry have respectively enabled the high-fidelity monitoring of (i) quasi-static deformation fields~\cite{sause2016digital,sanchez2008use}, and (ii) ultrasonic wave fields~\cite{tokmashev2013experimental,pourahmadian2018elastic}, on the surface of solid bodies. In turn, these capabilities have opened a door for the applications of elastography toward high-fidelity constitutive characterization of \emph{sheet-like} (or slab-like) laboratory specimens of either natural or engineered materials. In such an approach the in-plane kinematic data, captured across either side of the sheet, play the role of the full-field deformation maps consumed by two-dimensional (2D) elastography. Recently, it was demonstrated experimentally~\cite{pourahmadian2018elastic} that such sensory data are amenable to a constitutive interpretation within the framework of \emph{plane-stress elasticity}, provided that the in-plane length scales (e.g. the wavelength of ultrasonic motion) exceed the thickness of a specimen by a ``decade'' -- implying a factor of 4-5 or higher. 

In linear elasticity, well-known expressions relating the (native) elastic moduli of a 3D solid to their plane-stress equivalents~\cite{Kolsky,Malvern} provide a direct means for interpreting the ``raw'' results of 2D plane-stress elastography. When considering the like characterization of dissipative solids, on the other hand, the situation is far less clear. For instance, efforts have been made to infer the 2D stress field from the in-plane strain measurements in sheet-like viscoelastic specimens~\cite{yoneyama1999evaluation,taguchi2020computing}, which inherently assumes prior knowledge of the plane-stress viscoelastic moduli. In a similar vein, the authors in~\cite{comitti2024thermomechanical} pursued mechanical characterization of homogeneous viscoelastic membranes with the aid of a DIC system. In their study, membrane specimens were characterized through the prism of 2D  viscoelasticity, endowed with an orthotropic creep compliance matrix expanded in Prony series. Unfortunately, no effort was made to relate the resulting plane-stress compliances and relaxation times to the native (i.e.~``bulk") properties of a 3D viscoelastic solid.  

To help bridge the gap, our aim is to elucidate the constitutive behavior of linear viscoelastic solids under the plane stress condition. To this end we first revisit in Section~\ref{pspsx} the venerable concept of plane stress elasticity from the viewpoint of \emph{projection operators}, which caters for a compact treatment of the viscoelastic problem. For an in-depth analysis, in Section~\ref{ps:visco} we focus on the frequency-domain behavior by resorting to the Fourier transform and the mathematical framework of \emph{tempered distributions}, which extends the Fourier analysis to functions (prominent in linear viscoelasticity) for which Fourier integrals may not be convergent. In terms of the native 3D constitutive behavior, we consider the genus of viscoelastic models described by a linear constant-coefficient differential equation (whose Fourier image is hence polynomial), capable of describing (linear) rheological behavior affiliated with a generic arrangement of ``springs'' and ``dashpots''. With such premise, our analysis proceeds with the focus on isotropic viscoelastic solids -- whose bulk and shear modulus each derive from a suite of classical ``polynomial'' models. For generality, we also examine in Section~\ref{GSM} the reduction of \emph{thermodynamic potentials} describing linear viscoelasticity under the plane stress condition, paving the way for future applications of the error-in-constitutive-relation approach~\cite{jmps2024} to material characterization of sheet-like solid specimens. We consider a particular class of thermodynamic models, known as generalized standard materials (GSM), that are described by the Helmholtz free energy and dissipation potential densities. Loosely speaking, we find that the plane-stress reductions of the fourth-order elasticity and dissipation tensors (featured by the GSM description) follow the ``\emph{generalized-inverse-of-the-projected-inverse}'' pattern featured by the preceding analysis. 
Bulk memory functions are next derived in Section~\ref{A3} for the plane-stress forms of several cases of isotropic materials whose three-dimensional bulk and shear responses are described by classical scalar viscoelastic models. The study is concluded in Section~\ref{sec:numex} by a set of examples, illustrating the effect on the plane stress condition on the behavior of both isotropic and anisotropic viscoelastic solids.

\section{Plane strain and plane stress in linear elasticity}
\label{pspsx}

\subsection{Preliminaries} In what follows we let
\[
i,j,k,l=\overline{1,3}, \qquad \matp,\matq,\matr,\mats=\overline{1,2},
\]
and refer the analysis to a Cartesian
coordinate system endowed with an orthonormal basis~$\{\bfe_i\}$. On adopting the Einstein summation convention, we introduce the position vectors $\bfx=\xi_i\hh\bfe_i\!\in\!\mathbb{R}^3$, $\bfxi=\xi_\matp\hh\bfe_\matp\in\mathbb{R}^2$, and write the fourth-order elasticity tensor as $\bfC=C_{ijkl}\exs\bfe_i\otimes\bfe_j\otimes\bfe_k\otimes\bfe_l$. Letting $\bfu=u_i \hh \bfe_i$, $\bfeps\!=\!\eps_{ij} \hh\hh \bfe_i\otimes\bfe_j=\tfrac{1}{2}(\nabla_{\!\bfx}\bfu+\nabla_{\!\bfx}\Tsup\bfu)$ and $\bfsig\!=\!\sigma_{ij} \hh\hh \bfe_i\otimes\bfe_j=\bfC\dip\bfeps$ denote respectively the germane displacement field, linearized strain field, and Cauchy stress field, we recall the quintessential two-dimensional approximations
\begin{equation} \label{ps1}
\begin{aligned}
\text{Plane strain:} & \quad \bfu(\bfx) = u_\matp(\bfxi)\hh\bfe_\matp  &\implies& 
\quad \eps_{j3} = \eps_{3j} = 0 \\
\text{Plane stress:} & \quad \sigma_{j3} = \sigma_{3j} = 0 &\stackrel{(\star)}{\implies}& 
\quad\bfsig(\bfx) := \sigma_{\matp\matq}(\bfxi)\hh \bfe_\matp \otimes \bfe_\matq 
\end{aligned} 
\end{equation}
referred to the $\bfxi$-plane, where the plane-stress inference ($\star$) is a customary simplifying assumption made possible by the balance of linear momentum $\nabla_{\!\bfx}\sip\hh\bfsig = \bfze$ being reduced to $\nabla_{\!\bfxi}\sip\hh\bfsig = \bfze$. For completeness, we note that the latter argument holds equally in the presence of body forces $\boldsymbol{f}(\bfx) := f_\matp(\bfxi)\hh \bfe_\matp$.    

\subsection{Projection operators}

Plane-strain and plane-stress 2D constitutive models each relate the in-plane stress tensor and corresponding strain tensor. With this in mind,
we introduce the fourth-order tensors
\begin{equation} \label{newten1}
\begin{aligned}
\boldsymbol{I} \,=\, \delta_{ik}\delta_{jl} \, \bfe_i\otimes\bfe_j\otimes\bfe_k\otimes\bfe_l, \qquad
\bfP \,=\,  \delta_{\matp\matr}\delta_{\matq\mats} \, \bfe_\matp\otimes\bfe_\matq\otimes\bfe_\matr\otimes\bfe_\mats,  \qquad
\bfR \,=\, \boldsymbol{I} - \bfP,
\end{aligned}
\end{equation}
with~$\bfI$ being the fourth-order identity tensor. Clearly, $\bfI\dip\bfP = \bfP\dip\bfI =\bfP$, $\bfI\dip\bfR=\bfR$, and $\bfR \dip\bfI = (\bfI\!-\!\bfP):\bfI = \bfR$. We also note that~$\bfP$ and~$\bfR$ verify
\begin{equation} \label{proj1}
\begin{aligned}
  \bfP\dip\bfP \,=\, \bfP, \qquad
  \bfP\dip\bfR = \bfR\hh\dip\bfP  = \bfze, \qquad
  \bfR\dip\bfR = \bfR\dip(\bfI-\bfP) = \bfR,
\end{aligned}
\end{equation}
the first and third equalities being characteristic properties of projection operators.
The definitions in~\eqref{newten1} allow us to write
\begin{equation} \label{deco1}
\bfsig = \bfsig\psup + \bfsig\rsup, \qquad
\bfeps = \bfeps\psup + \bfeps\rsup,
\end{equation}
where
\begin{equation} \label{ps2}
\begin{aligned}
\bfsig\psup &= \bfP \dip \bfsig = \sigma_{\matp\matq}\, \bfe_\matp\otimes\bfe_\matq, &\qquad \bfsig\rsup &= \bfR \dip \bfsig = \sigma_{\matp 3}\, (\bfe_\matp\otimes\bfe_3 + \bfe_3\otimes\bfe_\matp) + \sigma_{33}\, \bfe_3\otimes\bfe_3, \\
\bfeps\psup &= \bfP \dip \bfeps =  \eps_{\matp\matq}\, \bfe_\matp\otimes\bfe_\matq, &\qquad \bfeps\rsup &= \bfR \dip \bfeps = \eps_{\matp 3}\, (\bfe_\matp\otimes\bfe_3 + \bfe_3\otimes\bfe_\matp) + \eps_{33}\, \bfe_3\otimes\bfe_3,
\end{aligned}
\end{equation}
making use of the facts that $\sigma_{\matp 3}=\sigma_{3\matp}$ and $\eps_{\matp 3}=\eps_{3\matp}$. Recalling~\eqref{ps1}, we find that
\begin{equation}\label{ps3}
\begin{aligned}
\bfeps\rsup &= \bfze, &\quad \bfeps &= \bfeps\psup = \bfP\dip\bfeps = \bfP\dip\bfeps\psup   &\quad& \text{(plane strain)}, \\
\bfsig\rsup &= \bfze, &\quad \bfsig &= \bfsig\psup = \bfP\dip\bfsig = \bfP\dip\bfsig\psup  &\quad& \text{(plane stress)}.
\end{aligned}
\end{equation}

\subsection{Plane strain}

Thanks to~\eqref{ps1}, \eqref{newten1}, \eqref{deco1} and~\eqref{ps3}, we have
\begin{equation}
\bfsig = \bfC \dip \bfeps
= \bfC \dip \bfP \dip \bfeps\psup  \qquad \Rightarrow \qquad
\bfsig\psup = \bfC^{\eps}\dip \bfeps\psup,
\end{equation}
where
\begin{equation} \label{strain1}
\bfC^{\exs\eps} \,=\, \bfP \dip\bfC \dip \bfP \,=\,
C_{\matp\matq\matr\mats} \, \bfe_\matp\otimes\bfe_\matq\otimes\bfe_\matr\otimes\bfe_\mats,
\end{equation}
which recovers the expected result in that the plain-strain moduli equal their ``bulk" counterparts, i.e. $C^{\exs\eps}_{\matp\matq\matr\mats} = C_{\matp\matq\matr\mats}$. For completeness, we also note that $\bfsig\rsup = \bfR\dip\bfC\dip\bfP\dip\bfsig\psup$.

\subsection{Plane stress} \label{psproj}

In this case, from~\eqref{ps1}--\eqref{ps3} we find that
\begin{equation}\label{stress2}
\bfP\dip\bfsig\psup = \bfsig\psup \,=\, \bfsig \,=\, \bfC \dip (\bfeps\psup + \bfeps\rsup)
 = \bfC \dip\bfeps\psup + \bfC\dip\bfR\dip\bfeps,
\end{equation}
and hence
\begin{equation}
  \bfP\dip\bfC^{-1}\dip\bfP\dip\bfsig\psup = \bfP\dip\bfeps\psup + \bfP\dip\bfR\dip\bfeps
 = \bfP\dip\bfeps\psup = \bfeps\psup.
\end{equation}
Solving for~$\bfsig\psup$, we accordingly obtain
\begin{equation}\label{stress3}
\bfsig\psup = \bfC^{\exs\sigma} \dip \bfeps\psup,
\end{equation}
where
\begin{equation}\label{stress4}
\bfC^{\exs\sigma} \,=\, \big(\bfP\dip\bfC^{-1}\dip\bfP \big)^{\dagger} \,:=\, C_{\matp\matq\matr\mats}^{\exs\sigma} \, \bfe_\matp\otimes\bfe_\matq\otimes\bfe_\matr\otimes\bfe_\mats.
\end{equation}

\begin{remark} \label{MP}
For any fourth-order tensor $\bfA$ endowed with major and minor symmetries, $\big(\bfP\dip\bfA\dip\bfP\big)^{\dagger}$ in~\eqref{stress3}--\eqref{stress4} denotes the Moore-Penrose pseudoinverse of the projection~$\bfP\dip\bfA\dip\bfP$ that accounts for (a portion of) its null space spanned by $\{\bfe_i\otimes\bfe_j\otimes\bfe_k\otimes\bfe_l\}\setminus\{\bfe_\matp\otimes\bfe_\matq\otimes\bfe_\matr\otimes\bfe_\mats\}$. More specifically, $\big(\bfP\dip\bfA\dip\bfP\big)^{\dagger}$ amounts to computing the usual inverse $\big(\bfP\dip\bfA\dip\bfP\big)^{-1}$ in the reduced basis $\{\bfe_\matp\otimes\bfe_\matq\otimes\bfe_\matr\otimes\bfe_\mats\}$, and then restating the result in $\{\bfe_i\otimes\bfe_j\otimes\bfe_k\otimes\bfe_l\}$.
Practical details on how to compute $\bfC^{\exs\sigma}$ within the framework of matrix algebra (using the so-called Voigt or contracted notation~\cite{voigt1910lehrbuch,ting:96}) are provided in Appendix~\ref{A:Voigt}. 
\end{remark}

For completeness, we recall that the plane-stress assumption generally violates the kinematic compatibility relations~\cite[e.g.][]{Malvern} in that $\nabla\!\times \bfeps(\bfx)\times\! \nabla\Tsup \!\neq\! \bfze$, where $\nabla\times$ is the curl operator. Specifically, since the second of~\eqref{ps1} implies $\bfeps(\bfx)\!=\!\bfeps(\bfxi)$, the featured compatibility relations include the requirements  
\[
\frac{\partial^2\eps_{33}}{\partial\xi_\matp \partial\xi_\matq} \,=\, C_{33\matp\matq}\, \frac{\partial^2\sigma_{\matp\matq}}{\partial\xi_\matp \partial\xi_\matq} \,=\, 0,  
\]
that are met only for particular stress distributions where the linear combination $C_{33\matp\matq}\sigma_{\matp\matq}(\bfxi)$ is an affine function of $\xi_1$ and~$\xi_2$. Such possible compatibility violations are usually acceptably ignored for solids that are ``thin'' (relative to other characteristic length scales in the problem) in the out-of-plane direction.

\section{Plane stress in linear viscoelasticity}
\label{ps:visco}

\noindent Hereon, our focus is on linear viscoelastic solids whose states have quiescent past until time $t\!=\!0$. Their response will accordingly be described via causal functions (i.e. functions $f$ such that $f(t)\!=\!0$ for~$t\!\leqslant\!0$) and causal distributions (whose evaluation on any smooth test function supported in $\Rbb_{<0}$ vanishes). On recalling the convolution product
\[
f(t)\star g(t) \,:=\, \int_{-\infty}^{\infty} f(\tau)\hh g(t-\tau) \dd\tau \,=\, \int_{-\infty}^{\infty} f(t-\tau)\hh g(\tau) \dd\tau, 
\]
for such a generic viscoelastic solid we have
\begin{equation} \label{visco1}
  \bfsig(\bfx,t) \,=\, \bfE(t) \star \dot{\bfeps}(\bfx,t), 
\end{equation}
where $\dot{f}\!=\!\partial{f}/\partial{t}$; $\bfE(t)$ is a fourth-order tensor, composed of causal \emph{stress relaxation functions} $E_{ijkl}(t)$, endowed with major and minor symmetries; and the convolution symbol ``$\star$'' implies requisite tensor contraction(s) whenever appropriate. Taking advantage of the blanket causality assumption, we obtain
\begin{equation} \label{convodef}
\bfE(t) \star \dot{\bfeps}(\bfx,t)
 = \int_{0}^{t} \bfE(t\!-\!\tau)  \hspace*{-0.4pt}: \dot{\bfeps}(\bfx,\tau)\dtau     
\end{equation}
when the time integral is finite. However, as will be demonstrated later, viscoelastic relationships often involve variables or time-dependent moduli that make the above integral ill-defined due to insufficient decay at infinity, singularities inherent to~$\bfE$, or strain discontinuities. Similarly, we will need to handle Fourier transforms whose integrals are not well defined in the ordinary sense. To help manage such impediments, we pursue the analysis within the framework of classical distribution theory for the time variable, which allows all relevant quantities and operations to be well-defined while preserving generality and maintaining conciseness. For example, convolutions such as~\eqref{visco1} are meaningful for all pairs of causal distributions, and yield causal distributions, even when their integral forms break down. Amid an abundant literature, we refer to the lecture notes~\cite{salo:13} which provide all precise definitions and statements pertaining to distribution theory, convolution and Fourier analysis that address our needs.\enlargethispage*{3ex}

In many applications, \eqref{visco1} is conveniently restated (using the identity $(u\star v)^{\DOT}=u\star \dot{v} = \dot{u}\star v$ for any pair of convolvable distributions $\{u,v\}$) as
\begin{equation} \label{visco1cm}
  \bfsig(\bfx,t) = \big( \bfE(t) \star \bfeps(\bfx,t) \big)^{\DOT} = \dot{\bfE}(t) \star \bfeps(\bfx,t) := \bfC(t) \star \bfeps(\bfx,t),
\end{equation}
where $\bfC(t)\!=\!\dot{\bfE}(t)$ is the so-called \emph{tensor memory function}, or memory tensor for short. We will use the memory tensor as a starting description of the 3D constitutive behavior of viscoelastic solids for obtaining the respective plane-stress constitutive models, and then deduce the plane-stress relaxation tensor accordingly.

The inverse form of~\eqref{visco1cm} is often formulated as
\begin{equation} \label{visco2}
  \bfeps(\bfx,t) \,=\, \bfJ(t) \star \dot{\bfsig}(\bfx,t),
\end{equation}
where $\bfJ(t)$ is a fourth-order tensor composed of causal \emph{creep compliance functions} $J_{ijkl}(t)$ that carries major and minor symmetries.  To emphasize their physical meaning, $\bfE(t)$ and $\bfJ(t)$ are illustrated in Fig.~\ref{relax1}. To facilitate the analysis, viscoelastic relationship~\eqref{visco2} can be conveniently rewritten as 
\begin{equation} \label{visco:inv}
  \bfeps(\bfx,t) \,=\, \bfD(t)\star\bfsig(\bfx,t),
\end{equation}
where the tensor function $\bfD(t):=\,\dot{\!\bfJ}(t)$ can be interpreted as the (convolutional) inverse memory tensor. Indeed, combining~\eqref{visco1cm} and~\eqref{visco:inv}, we must have
\begin{equation}
  \bfeps(\bfx,\dotp) \,=\, \bfD\star \lpar \bfC\star\bfeps(\bfx,\dotp) \rpar
   \,=\, \lpar \bfD\star\bfC \rpar \star\bfeps(\bfx,\dotp) \qquad\implies\qquad \bfD(t)\star\bfC(t)=\delta(t)\hh \IS, \label{CD=I}
\end{equation}
(the convolution operation being associative for any triple of causal distributions, their supports being a convolvable family) where $\IS=\tdemi(\delta_{ik}\delta_{jl}+\delta_{il}\delta _{jk})\,\bfe_i\!\otimes\!\bfe_j\!\otimes\!\bfe_k\!\otimes\!\bfe_l$ is the symmetric fourth-order identity tensor.\enlargethispage*{3ex}

\begin{figure}[ht]
\centering{\includegraphics[width=0.75\linewidth]{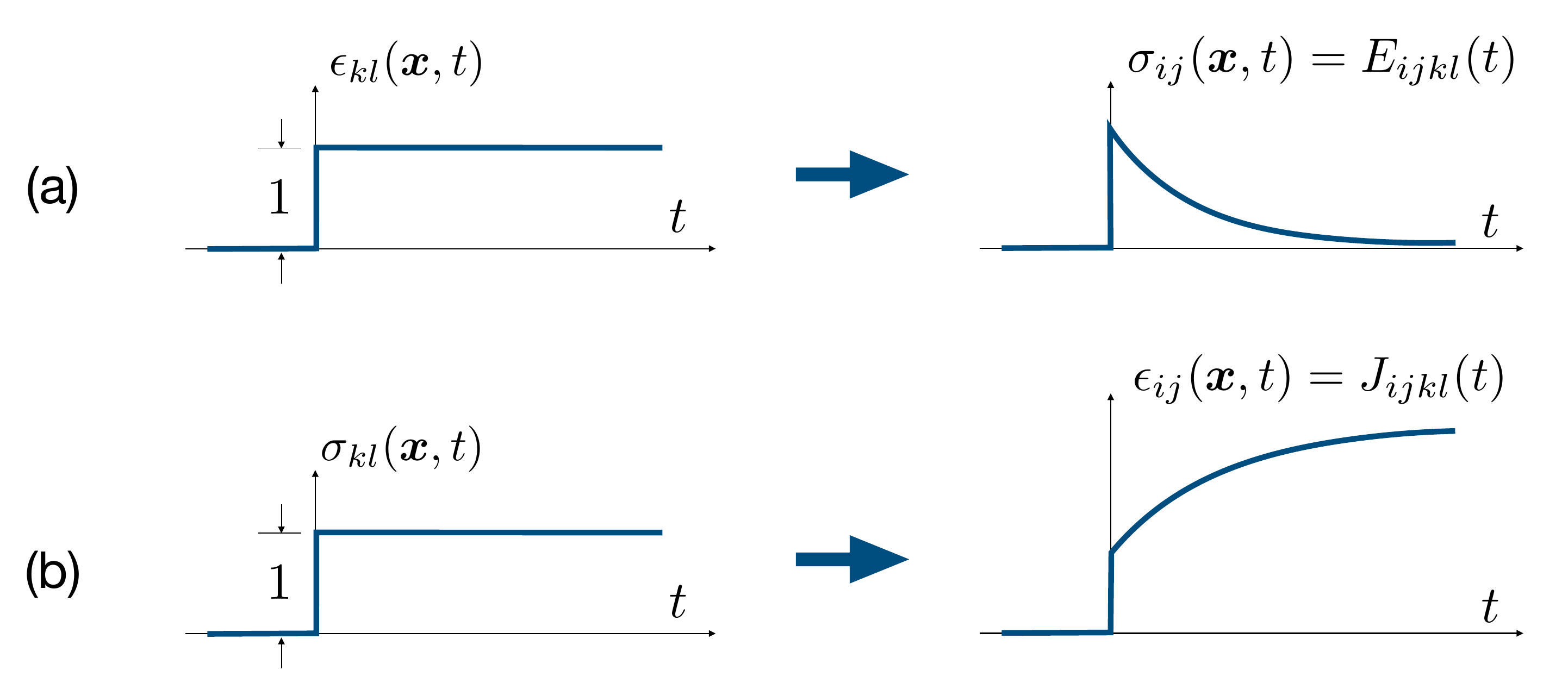}}
\caption{Schematics of (a) stress relaxation function $E_{ijkl}(t)$ assuming
$\bfeps(\bfx,t) = H(t)\hh \bfe_k\nes\otimes\nes\bfe_l$
and (b) creep function $J_{ijkl}(t)$ assuming
$\bfsig(\bfx,t) = H(t)\hh \bfe_k\!\otimes\bfe_l$.} \label{relax1}
\end{figure}

On recalling~\eqref{ps1} with an added temporal variable, namely
\begin{equation} \label{ps1bisx}
\begin{aligned}
\text{Plane strain:} & \qquad \bfu(\bfx,t) = u_\matp(\bfxi,t)\hh\bfe_\matp &\implies&
\quad   \eps_{j3}(\bdot,t) = \eps_{3j}(\bdot,t) = 0 \\
\text{Plane stress:} & \qquad \sigma_{j3}(\bdot,t) = \sigma_{3j}(\bdot,t) = 0 &\stackrel{(\star)}{\implies}&
\quad \bfsig(\bfx,t) := \sigma_{\matp\matq}(\bfxi,t)\hh \bfe_\matp \otimes \bfe_\matq  
\end{aligned}
\end{equation}
and following the analysis presented for the elastic case, from~\eqref{visco1cm} and~\eqref{visco:inv} we immediately obtain
\begin{equation} \label{ps1bis}
\begin{aligned}
  & \text{Plane strain:} \quad
  & \bfsig\psup(\bfxi,t) &= \bfC^{\exs\eps}(t) \star \bfeps\psup(\bfxi,t),
  & \quad\bfC^{\exs\eps}(t) &:=\, \bfP \dip\bfC(t) \dip \bfP \label{visco3} \\[1ex]
  & \text{Plane stress:} \quad
  & \bfeps\psup(\bfxi,t) &= \bfD^{\sigma}(t) \star \bfsig\psup(\bfxi,t),
  & \quad \bfD^{\sigma}(t) &:=\, \bfP \dip\bfD(t) \dip \bfP \label{visco4}
\end{aligned}~. 
\end{equation}
By analogy to the comment made earlier, we remark that the plane-stress inference~$(\star)$ in~\eqref{ps1bisx} is a simplifying assumption made possible by the balance of linear momentum $\nabla_{\!\bfx}\sip\hh\bfsig - \rho\hh \ddot\bfu = \bfze$  being reduced to $\nabla_{\!\bfxi}\sip\hh\bfsig\psup - \rho\hh \ddot{u}_\matp \hh \bfe_\matp= \bfze$ and  $\ddot{u}_3=0$.

\paragraph{Plane stress} From now on, our focus is on evaluating the viscoelastic plane-stress behavior in more depth. Assuming the plane stress condition at all times, we have
\[
  \bfsig\psup(\bfx,t)
 = \bfP\dip\bfsig\psup(\bfx,t)
 = \bfsig(\bfx,t)
 = \bfC(t)\star \bfeps(\bfx,t).
\]
Then, taking the time convolution of the above equality with $\bfP\dip\bfD(t)$ from the left, we find
\[
\begin{aligned}
    \lpar \bfP\dip\bfD(t)\dip\bfP \rpar \star \bfsig\psup(\bfx,t)
 &= \bfP\dip\lpar (\bfD(t)\star \bfC(t) \rpar \star\lpar \bfP\dip\bfeps(\bfx,t) + \bfR\dip\bfeps(\bfx,t) \rpar \\
 &= \bfP\dip\lpar \delta(t)\hh  \star \bfeps(\bfx,t) \rpar
 = \bfeps\psup(\bfx,t). 
\end{aligned}
\]
As a result, the plane-stress memory tensor $\bfC^{\sigma}(t)$ satisfying
\begin{equation}
  \bfsig\psup(\bfx,t) = \bfC^{\sigma}(t)\star\bfeps\psup(\bfx,t) \label{ps:time}
\end{equation}
must be the convolutional inverse of $\bfP\dip\bfD(t)\dip\bfP$ in that
\begin{equation}
  \bfC^{\sigma}(t) \star \lpar \bfP\dip\bfD(t)\dip\bfP \rpar = \delta(t)\hh \IS. \label{DC:time}
\end{equation}
In practice, exploiting~\eqref{ps:time} relies on using either Laplace or Fourier transform, either choice taking advantage of the relevant statement of the convolution theorem that reduces convolutions to multiplications. In the context of practical applications (e.g.~material characterization), one may prefer to seek the Fourier transform counterparts of~\eqref{visco1cm} and~\eqref{visco:inv} owing to their relevance to time-harmonic deformation and the notion of complex viscoelastic moduli.

\subsection{Evaluation of~$\bfC^{\sigma}$ via the Fourier transform}

The Fourier transform and its inverse are given by
\begin{equation} \label{fourier}
\hat{f}(\omega) = \Fcal[f](\omega) := \int_{\mathbb{R}} f(t) \hh e^{\ii\omega t} \dt, \qquad\quad
f(t) = \Fcal^{-1}[\hat{f}](t) := \frac{1}{2\pi} \int_{\mathbb{R}} \hat{f}(\omega) \hh e^{-\ii\omega t} \dom
\end{equation}
when the functions $f$ or $\hat{f}$ are summable. The transform pair~\eqref{fourier} has well-defined extensions to the space $\Scal'$ of \emph{tempered distributions}, namely distributions defined as continuous linear functionals on the space $\Scal$ of all $C^{\infty}(\Rbb)$ test functions with faster-than-polynomial decay at infinity. The Fourier convolution theorem for distributions~\cite[e.g.][Thm.~3.9.9]{salo:13} then yields 
\begin{equation} \label{FCT}
\Fcal[u\star v](\omega) = [\hat{u} \dotp \hat{v}](\omega) \quad\text{or}\quad 
\mathcal{F}[u \dotp v](\omega) = \frac{1}{2\pi} [\hat{u} \star \hat{v}](\omega)  
\end{equation}
for any pair $(u,v)$ of causal tempered distributions, provided $\hat{u}\!=\!\Fcal[u]$ (say) belongs to the space $\Ocal_{\nes M}$ of $C^{\infty}(\Rbb)$ functions with slow (i.e. at most polynomial) growth at infinity, in which case $u$ belongs to the space $\Ocal'_{\nes C}$ of distributions with faster-than-polynomial decay at infinity.

We can now apply~\eqref{FCT} to the respective Fourier transforms of~\eqref{visco1cm} and~\eqref{CD=I}, obtaining 
\begin{equation}
  \hat{\bfsig}(\bfx,\omega) = \hat{\bfC}(\omega)\dip\hat{\bfeps}(\bfx,\omega), \qquad
  \hat{\bfD}(\omega)\dip\hat{\bfC}(\omega) = \IS,  \label{visco2cm}
\end{equation}
where $\hat{\bfC}(\omega)$ is a fourth-order tensor composed of \emph{complex modulus functions $\hat{C}_{ijkl}(\omega)$}. The same steps applied to~\eqref{ps:time} and~\eqref{DC:time} then yield
\begin{equation}
  \hat{\bfsig}\psup(\bfx,\omega)
 = \hCS(\omega)\dip\hat{\bfeps}\psup(\bfx,\omega), \qquad
  \hCS(\omega) = \lpar \bfP\dip\hat{\bfD}(\omega)\dip\bfP \rpar^{\dagger}, \label{visco4cm}
\end{equation}
where the tensor of complex compliances, $\hat{\bfD}(\omega)$, is obtained from the tensor of complex moduli $\hat{\bfC}(\omega)$ via the second of~\eqref{visco2cm}.

\subsection{Example genus of viscoelastic models} \label{models:EDO}

We next consider a scalar and local version of the constitutive relationship~\eqref{visco1cm} featuring stress $\sigma(t)$, strain $\eps(t)$ and a memory function $C(t)$ that will both illustrate the foregoing considerations and provide building blocks for applications to specific 3D models. As examined in e.g.~\cite{makris2019frequency}, such viscoelastic relationships are often described in terms of a phenomenological differential equation
\begin{equation} \label{scalarDE}
  \Pcal\Lpar\der{}{t}\Rpar\sigma
 = \Qcal\Lpar\der{}{t}\Rpar\eps \quad\text{where}\quad
  \Pcal\Lpar\der{}{t}\Rpar = \sum_{m=0}^M p_m\frac{\text{d}^m}{\text{d}t^m}, \quad
  \Qcal\Lpar\der{}{t}\Rpar = \sum_{n=0}^N q_n\, \frac{\text{d}^n}{\text{d}t^n} \quad (p_m, q_n\!\in\Rbb)
\end{equation}
that is affiliated with a given arrangement of ``springs'' and ``dashpots'' \cite{findley2013creep,tschoegl}. 

\subsubsection{Memory function and its convolutional inverse} \label{meminv}

On applying the Fourier transform~\eqref{fourier} to~\eqref{scalarDE}, the complex modulus function and its inverse (i.e.~the complex compliance function) are obtained as
\begin{equation} \label{poles1}
  \hat{C}(\omega) = \frac{\Qcal(-\ii\omega)}{\Pcal(-\ii\omega)}
  = \frac{\sum_{n=0}^N q_n (-\ii\omega)^{n}}{\sum_{m=0}^M p_m (-\ii\omega)^m}, \qquad~
  \hat{D}(\omega) = \frac{\Pcal(-\ii\omega)}{\Qcal(-\ii\omega)}
  = \frac{\sum_{m=0}^M p_m (-\ii\omega)^m}{\sum_{n=0}^N q_n (-\ii\omega)^n}. 
\end{equation}
As discussed in~\cite[Chap.~3]{tschoegl}, the zeros~$\omega_i$ ($i\!=\!\overline{1,M}$) of the polynomial $\oo\mapsto\Pcal(-\ii\oo)$ for such models are purely imaginary with $\Im(\oo_i)<0$ (in particular $\Pcal(0)\!\neq\! 0$), and the same argument applies to $\oo\mapsto\Qcal(-\ii\oo)$ except that $\oo_i\!=\!0$ as a simple zero is also possible. In addition to their role in establishing the causality of $C(t)$ and $D(t)$ (see Section~\ref{KK}), these properties imply that the rational function $\oo\mapsto\hat{C}(\omega)$ has no poles on the real axis; it is thus $C^{\infty}(\Rbb)$ and, while often not summable over $\Rbb$, belongs to the space $\Ocal_{\nes M}$. On the other hand, $\oo\!=\!0$ may be a simple pole of $\hat{D}(\omega)$, which is then to be defined as a principal value; the Maxwell model (see Table~\ref{tab1}(b)) is an example of such behavior. Then, properties such as the distributional equality $\omega\,\text{PV}\tfrac{1}{\omega}\!=\!1$ (the second multiplier denoting the principal value of~$\omega^{-1}$) ensure that $\hat{D}(\omega)\hat{C}(\omega)\!=\!1$ holds as an equality between distributions, so that $\hat{D}(\omega)$ can be interpreted as a complex compliance.

The memory function $C(t)$ and its inverse $D(t)$ are then to be obtained by taking the inverse Fourier transform of $\hat{C}(\omega)$ and $\hat{D}(\omega)$. Irrespective of the orders $M$ and~$N$ of the differential polynomials $\Pcal$ and~$\Qcal$, integral form~\eqref{fourier} of the inverse Fourier transform will fail to be convergent for at least one of $\{C(t),D(t)\}$. Yet, the featured Fourier transforms are still well-defined as operations on distributions, and readily computable with the aid of known Fourier transform pairs that include generalized functions such as the Dirac delta function (and its derivatives) or principal values \cite{makris2019frequency,Evans2009}. In this setting, every $\hat{C}(\omega)$ and $\hat{D}(\omega)$ arising from the class~\eqref{scalarDE} of phenomenological viscoelastic models can be shown, with the help of the Kramers-Kronig relations (see Section~\ref{KK}), to be the Fourier image of a causal distribution.


\subsubsection{Relaxation and creep functions} \label{relcreep}

Thanks to the above properties, the Fourier image of the relaxation function $E(t)$ due to~\eqref{scalarDE} can be conveniently computed from the equality $\sigma(t)=(E\star\dot{\eps})(t)=(C\star\eps)(t)$ with $\eps(t)=H(t)$, which yields the distributional relation
\begin{equation} \label{Eomega}
  \hat{E}(\omega)
 = \Big(\frac{\ii}{\omega} + \pi\delta(\omega) \Big)\, \hat{C}(\omega)
 = \Big(\frac{\ii}{\omega} + \pi\delta(\omega) \Big)\,\frac{\Qcal(-\ii\omega)}{\Pcal(-\ii\omega)}
 = \frac{\ii}{\omega}\frac{\Qcal(-\ii\omega)}{\Pcal(-\ii\omega)} + \pi \frac{\Qcal(0)}{\Pcal(0)} \hh \delta(\omega)
\end{equation}
The above-discussed properties of the polynomials $\oo\mapsto\Pcal(-\ii\oo)$ and $\oo\mapsto\Qcal(-\ii\oo)$ imply that either: (i) $\hat{E}(\omega)$ has a $\omega^{-1}$ singularity at the origin (when $\Qcal(0)\shneq0$), or (ii) $\hat{E}(\omega)$ has no singularity (when $\Qcal(0)\sheq0$ with single multiplicity). In case (i), the creep function $J(t)$ can be obtained in the same way (applying the Fourier transform to the equality $\eps(t)=(J\star\dot{\sigma})(t)=(D\star\sigma)(t)$ with $\eps(t)=H(t)$), to obtain
\begin{equation}
  \hat{J}(\omega)
 = \Big(\frac{\ii}{\omega} + \pi\delta(\omega) \Big)\, \hat{D}(\omega)
 = \Big(\frac{\ii}{\omega} + \pi\delta(\omega) \Big)\,\frac{\Pcal(-\ii\omega)}{\Qcal(-\ii\omega)}
 = \frac{\ii}{\omega}\frac{\Pcal(-\ii\omega)}{\Qcal(-\ii\omega)} + \pi\frac{\Pcal(0)}{\Qcal(0)}\hh \delta(\omega)
\end{equation}
In case (ii), this argument no longer holds since $\delta(\omega)\hat{D}(\omega)$ loses meaning as a distribution due to the singularity of $\hat{D}(\omega)$, indicating that the Fourier convolution theorem is not applicable to the pair $\{D(t),\sigma(t)\}$. Instead, $J(t)$ is sought from the equality $\dot{\sigma}(t)= (\dot{C}\star\eps)(t) = \big(\dot{C}\star(J\star\dot{\sigma})\big)(t)$, to which the Fourier convolution theorem applies and yields
\begin{equation}
  1 = -\ii\omega\hh \hat{C}(\omega)\hat{J}(\omega) = -\ii\omega^2\frac{\Rcal(-\ii\omega)}{\Pcal(-\ii\omega)}\hat{J}(\omega) \qquad\text{i.e.}\quad -\ii\omega^2 \hat{J}(\omega) = \frac{\Pcal(-\ii\omega)}{\Rcal(-\ii\omega)}, 
\end{equation}
having set $\Qcal(-\ii\omega):=\omega\hh \Rcal(-\ii\omega)$ where the polynomial $\omega\mapsto\Rcal(-\ii\omega)$ has no real zeros. The solution $\hat{J}(\omega)$ of the above distributional equality that defines a causal creep function $J$ is given by
\begin{equation}
\hat{J}(\omega) = \frac{\ii}{\omega}\frac{\Pcal(-\ii\omega)}{\Qcal(-\ii\omega)} + 
\alpha\hh \delta(\omega) + \beta\hh \delta'(\omega), \qquad
\alpha = \pi\hh f'(0), \quad \beta = -\hh \pi \hh f(0), \quad f(\omega) := \frac {\Pcal(-\ii\omega)}{\Rcal\cal(-\ii\omega)} 
\label{Jhat:Dsing}
\end{equation}
where the values of constants $\alpha$ and~$\beta$ are obtained by enforcing the Kramers-Kronig relations~\eqref{visco10cm}, satisfied by every causal function, in terms of $\hat{J}(\omega)$ (see Appendix~\ref{A4} and Appendix~\ref{alfabeta} for details). The distributional equality 
\begin{equation} \label{EJCDprod}
-\omega^2 \hat{E}(\omega)\hat{J}(\omega) = 1
\end{equation}
then holds in both cases (i) and (ii), aided by the distributional property $\omega^2\hh \text{FP}\omega^{-2}=1$ of the finite part (FP) of $\omega^{-2}$ in case (ii). 


\begin{figure}[b]
\centering{\includegraphics[width=0.95\linewidth]{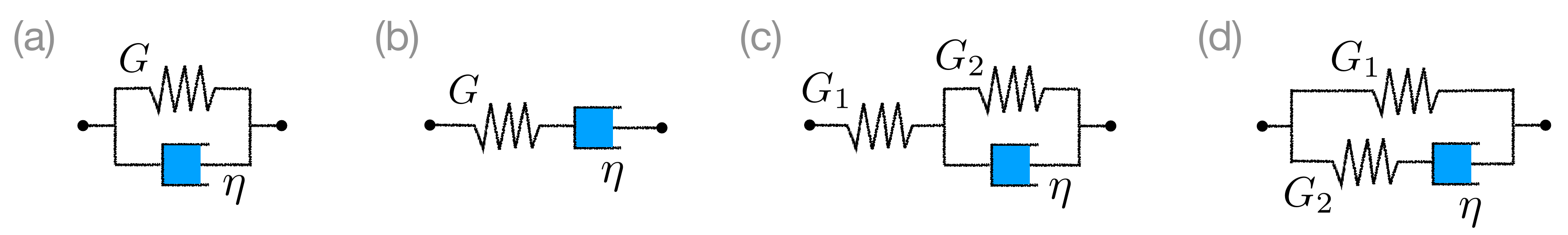}}
\caption{Elementary mechanical models of the viscoelastic behavior: (a) Kelvin-Voigt solid, (b) Maxwell material, (c) standard linear solid, and (d) Zener model. Note that (a) and (d) are characterized by $M\!<\!N$, while (b)--(c) feature $M\!=\!N$.} \label{viscomod1}
\end{figure}

To illustrate the foregoing considerations, Table~\ref{tab1} lists the Fourier transform pairs \cite{makris2019frequency} for the memory, creep and relaxation functions characterizing the viscoelastic models shown in Fig.~\ref{viscomod1}, see Appendix~\ref{A:FTP} for supporting details. In the table, $H(\dotp)$ is the step function; $\delta'(\dotp)$ is the derivative of the Dirac distribution $\delta(\dotp)$; and
\begin{equation}\label{fourx11}
\begin{aligned}
\text{(a,b):} \quad \gamma=G, \quad \tau &= \frac{\eta}{G}; \quad &
\text{(c):} \quad \gamma &= \frac{G_1 G_2}{G_1\!+G_2}, \quad  & \tau_1 &= \, \frac{\eta}{G_1\!+G_2}, \quad & \tau_2 &= \frac{\eta}{G_2}; \ \\
&&~\text{(d):} \quad \gamma &= G_1, & \tau_1  &= \,\frac{\eta}{G_2}, & \tau_2 &= \eta\hh \frac{G_1\!+G_2}{G_1 G_2}.
\end{aligned}
\end{equation}
\begin{table}[t]
\centering \renewcommand{\arraystretch}{1.35}
\begin{tabular}{|c||c|c|c|} \hline
Model & (a) & (b) & (c-d) \\ \hline\hline
$C(t)$  & $\gamma[\tau\hh \delta'(t) + \delta(t)]$ & $\gamma[\delta(t) - \frac{1}{\tau} \hh H(t) \hh e^{-t/\tau}]$ & $\gamma\hh[\frac{\tau_2}{\tau_1} \delta(t) - \frac{\tau_2-\tau_1}{\tau_1^2} \hh H(t) \hh e^{-t/\tau_1}]$ \\ \hline
$\hat{C}(\omega)$  & $\gamma [1 - \ii\tau\hh\omega]$ & $\gamma \hh \frac{\tau\hh \omega}{\ii+\tau\hh\omega}$ & $\gamma \, \frac{\ii+\tau_2\hh\omega}{\ii+\tau_1\hh\omega}$  \\ \hline\hline$J(t)$ & $\gamma^{-1}\hh H(t) [1 \!-\! e^{-t/\tau}]$ & $\gamma^{-1} H(t)[1 \!+\! t/\tau]$ & $\gamma^{-1}\hh H(t) [1 - \frac{\tau_2-\tau_1}{\tau_2} \hh e^{-t/\tau_2}]$  \\ \hline
$\hat{J}(\omega)$  & $\gamma^{-1} [\pi \delta(\omega) + \frac{\ii}{\omega} - \frac{\tau}{1-\ii\hh\tau\hh\omega}]$ &  $\gamma^{-1} [\pi \delta(\omega) + \frac{\ii}{\omega} - \frac{1}{\tau}( \frac{1}{\omega^2}+\ii\pi \delta'(\omega) )]$ & $\gamma^{-1} [\pi \delta(\omega) + \frac{\ii}{\omega} - \frac{\tau_2-\tau_1}{1-\ii\hh\tau_2\hh\omega}]$ \\ \hline\hline
$E(t)$  & $\gamma[\tau\hh \delta(t) + H(t)]$ & $\gamma H(t)\hh e^{-t/\tau}$ & $\gamma\hh H(t) [1 + \frac{\tau_2-\tau_1}{\tau_1}\hh e^{-t/\tau_1}]$  \\ \hline
$\,\hat{E}(\omega)$  & $\gamma\hh [\tau + \pi \delta(\omega) + \frac{\ii}{\omega}]$ & $\gamma\hh \frac{\tau}{1-\ii\hh\tau\hh\omega} $ & $\gamma \hh [\pi \delta(\omega) + \frac{\ii}{\omega} +\frac{\tau_2-\tau_1}{1-\ii\hh\tau_1\hh\omega}]$ \\ \hline\hline
\end{tabular} \renewcommand{\arraystretch}{1}
\caption{Memory, creep and relaxation functions: Fourier transform pairs for the viscoelastic models shown in Fig.~\ref{viscomod1}, with the elastic modulus $\gamma$ and characteristic times $\tau,\tau_1,\tau_2$ given by~\eqref{fourx11}.}
\label{tab1}
\end{table}

As can be seen from the display, only a single Fourier transform pair (namely $E(t)$ and $\hat{E}(\omega)$ for the Maxwell model) involves no generalized functions. Consistent with~\eqref{EJCDprod}, all basic models in Fig.~\ref{viscomod1} verify the equality $-\omega^2 \hh \hat{E}(\omega)\hat{J}(\omega) \nes=\nes 1$ between distributions. Recalling the discussion of $J(t)$ in Section~\ref{relcreep}, we observe that the Maxwell model is an example of case (ii), while all other models fall into category~(i).

\begin{remark} \label{char-relax}
For clarity of exposition, we hereon distinguish between two categories of timescales featured by a given viscoelastic model. Specifically, we discern between: (i) the characteristic times, defined as the timescales deriving intrinsically from the set of input model parameters, and (ii) the relaxation times, namely timescales featured by the exponentially-decaying contributions to the memory function, $C(t)$. For instance from Table~\eqref{tab1}, column (c-d) we observe that the standard linear solid (SLS) and Zener models each feature two characteristic times, $\{\tau_1,\tau_2\}$ given by~\eqref{fourx11}(c,d), and a single relaxation time equaling~$\tau_1$.
\end{remark}

\subsection{Isotropic materials}

For an isotropic viscoelastic solid, the memory tensor and its Fourier image (the tensor of complex moduli) read 
\[
  \hat{\bfC}(t) = 3\kaC(t)\JS + 2\muC(t)\KS, \qquad
  \hat{\bfC}(\omega) = 3\hkaC(\omega)\JS + 2\hmuC(\omega)\KS
\]
with the bulk $(\kaC)$ and shear $(\muC)$ memory functions assumed to be given by some chosen viscoelastic model, see Table~\ref{tab1} for examples. The fourth-order tensors $\JS$ and~$\KS$, representing respectively the volumetric and deviatoric components of an isotropic fourth-order tensor, are given by
\begin{equation} \label{JKS}
\JS = \tfrac{1}{3} \bfI_{\!2}\tens\bfI_{\!2}, \qquad \KS=\IS-\JS    
\end{equation}
where~$\bfI_{\!2}=\delta_{ij}\,\bfe_i\!\otimes\!\bfe_j$ is the second-order identity tensor. As $\JS$ and~$\KS$ verify the well-known relations
\begin{equation}\label{JKorth}
  \JS\dip\JS=\JS, \qquad \KS\dip\KS=\KS, \qquad \JS\dip\KS=\bfze,
\end{equation}
the inverse memory tensor $\bfD$ and its Fourier image must be of the form
\[
  \bfD(t) = 3\kaD(t)\JS + 2\muD(t)\KS, \qquad
  \hat{\bfD}(\omega) = 3\hkaD(\omega)\JS + 2\hmuD(\omega)\KS,
\]
where the complex compliances $\hkaD$ and $\hmuD$ are distributions required to satisfy the equalities
\[
  \hkaD(\omega) \hh \hkaC(\omega) = 1, \qquad \hmuD(\omega) \hh \hmuC(\omega) = 1.
\]
On introducing the reductions of~\eqref{JKS} given by
\begin{equation} \label{isovis4}
\begin{aligned}
\hh\JS\psupp &:= \tfrac{3}{2}\bfP\dip\JS\dip\bfP = \tfrac{1}{2}\hh \delta_{\matp\matq}\hh \delta_{\matr\mats} \, \bfe_\matp\otimes\bfe_\matq\otimes\bfe_\matr\otimes\bfe_\mats, \\
\hh\KS\psupp &:= -\tfrac{1}{3}\JS\psupp \nes+ \bfP\dip\KS\dip\bfP = -\hh\JS\psupp + \tfrac{1}{2}(\delta_{\matp\matr}\delta_{\matq\mats}+\delta_{\matp\mats}\delta_{\matq\matr}) \, \bfe_{\matp}\otimes\bfe_{\matq}\otimes\bfe_{\matr}\otimes\bfe_{\mats},
\end{aligned}
\end{equation}
the plane-stress projection of $\bfD$ is thus obtained in the frequency domain as
\[
  \bfP\dip\hat{\bfD}(\omega)\dip\bfP
  = \big[ \tfrac{2}{9}\hkaD(\omega) + \tfrac{1}{6}\hmuD(\omega) \big] \JS\psupp + \tfrac{1}{2}\hmuD(\omega)\KS\psupp.
\]
\begin{remark}
We note that the volumetric and deviatoric projections $\bfP\dip\JS\dip\bfP$ and $\bfP\dip\KS\dip\bfP$ do not enjoy the projection properties given by~\eqref{JKorth}; for instance $(\bfP\dip\JS\dip\bfP)\dip(\bfP\dip\JS\dip\bfP)\neq\bfP\dip\JS\dip\bfP$, since enforcing the equality would require the non-trivial entries of $\bfP\dip\JS\dip\bfP$ to equal $\tfrac{1}{2}$ in lieu of $\tfrac{1}{3}$. This motivates the above introduction of~$\JS\psupp$ and~$\KS\psupp$ which by design verify
\begin{equation}\label{JKorth-sig}
  \JS\psupp\dip\JS\psupp\!=\JS\psupp, \qquad \KS\psupp\dip\KS\psupp=\KS\psupp, \qquad \JS\psupp\dip\KS\psupp=\bfze.
\end{equation}
\end{remark}

Then, since
\[
\big(\alpha\hh\JS\psupp+\beta\hh\KS\psupp\big)^\dagger = \frac{1}{\alpha} \hh\JS\psupp + \frac{1}{\beta} \hh\KS\psupp
\]
due to~\eqref{JKorth-sig}, the tensor of complex moduli $\hCS$ relevant to the plane stress condition is obtained as
\begin{equation}\label{isovis7c}
  \hCS(\omega) = \lpar \bfP\dip\hat{\bfD}(\omega)\dip\bfP \rpar^{\dagger} = 2\hkaCS(\omega) \hh\JS\psupp + 2\hmuCS(\omega) \hh\KS\psupp,
\end{equation}
featuring the \emph{effective} (plane-stress) complex bulk and shear moduli
\begin{equation}\label{isovis8}
  \hkaCS(\omega) = \frac{9\hkaC(\omega)\hmuC(\omega)}{3\hkaC(\omega)+4\hmuC(\omega)}, \qquad
  \hmuCS(\omega) = \hmuC(\omega).
\end{equation}
When both $\hkaC$ and $\hmuC$ are rational functions arising from the class of phenomenological models described in Section~\ref{models:EDO}, namely $\hkaC(\oo)=\Qcal_{\kappa}(-\ii\oo)/\Pcal_{\kappa}(-\ii\oo)$ and $\hmuC(\oo)=\Qcal_{\mu}(-\ii\oo)/\Pcal_{\mu}(-\ii\oo)$, we obtain
\begin{equation} \label{hkaCS:ODE}
  \hkaCS(\omega) = \frac{9\Qcal_{\kappa}(-\ii\oo)\Qcal_{\mu}(-\ii\oo)}{3\Qcal_{\kappa}(-\ii\oo)\Pcal_{\mu}(-\ii\oo) + 4\Pcal_{\kappa}(-\ii\oo)\Qcal_{\mu}(-\ii\oo)}. 
\end{equation}
Here, the properties of the featured polynomial symbols (see Section~\ref{models:EDO}) imply that $\hkaCS(\omega)$ has no poles on the real axis, and is thus a $C^{\infty}(\Rbb)$ function with slow growth that may not be summable over~$\Rbb$. As demonstrated on several specific models in Section~\ref{A3}, the plane-stress memory function $\kaCS(t)$ is nonetheless recoverable by expanding~\eqref{hkaCS:ODE} in partial fractions and using several well-known (distributional) Fourier transform pairs. Moreover, the argument made in Section~\ref{KK} applies to~\eqref{hkaCS:ODE} and ensures the causal character of $\kaCS(t)$.

Viewed through the prism of a 1D mechanical analogue, $\hkaCS$ can be interpreted as the complex modulus of a serial ``spring-pot'' assembly comprising two spring-pot elements with respective moduli $\tfrac{9}{4}\hkaC$ and~$3\hmuC$, see Fig.~\ref{spring-pots}.  Formulas~\eqref{isovis8} also demonstrate that the plane stress condition affects (as expected) only the effective bulk modulus~$\hkaC^{\!\!\!\sigma}$, that generally features a mix of the characteristic times in bulk and shear.
\begin{figure}[h!]
\centering{\includegraphics[width=0.55\linewidth]{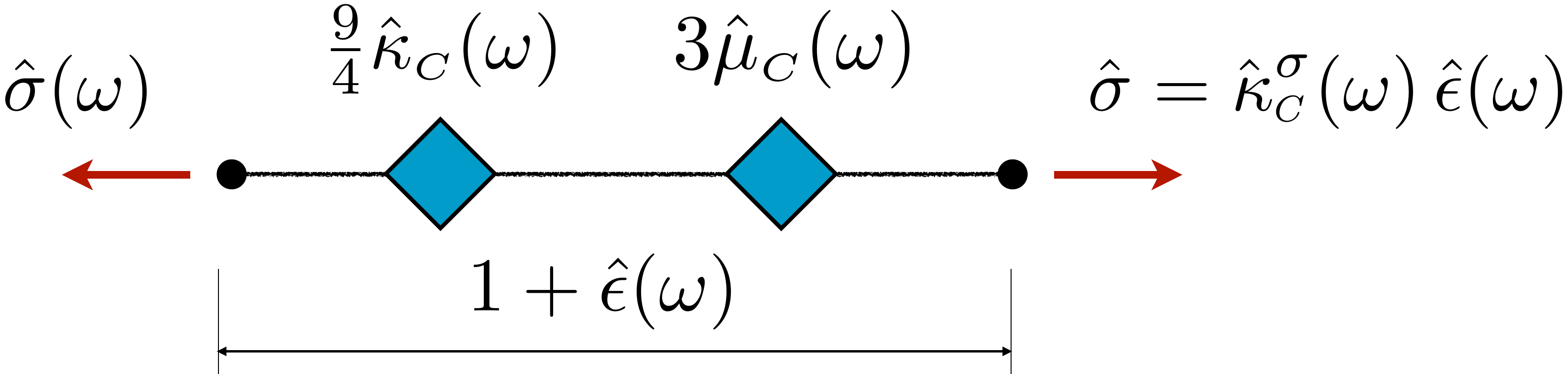}}
\caption{1D mechanical analogue depiction of the plane-stress bulk modulus, $\hkaCS(\omega)$.} 
\label{spring-pots}
\end{figure}

With~\eqref{isovis8} and~\eqref{hkaCS:ODE} in place, expression for the plane-stress memory tensor $\bfC^{\sigma}(t)$ is obtained by applying the inverse Fourier transform to~\eqref{isovis7c}, a result that features the effective (plane-stress) bulk and shear memory functions~$\kaC^\sigma(t)$ and~$\muC^\sigma(t)$.

\paragraph{Time-harmonic deformations}

For completeness, we note that the frequency-domain constitutive relationships~\eqref{visco2cm} and~\eqref{visco4cm} featuring the tensors of  complex moduli can be alternatively established in a phenomenological ``stand-alone'' fashion, without reference to the Fourier transform, by assuming time-harmonic deformations \cite[e.g.][]{findley2013creep} as
\[
\bfsig(\bfx,t) = \Re[\hat{\bfsig}(\bfx,\omega) \hh e^{-\ii\omega t}], \qquad \bfeps(\bfx,t) = \Re[\hat{\bfeps}(\bfx,\omega) \hh e^{-\ii\omega t}].
\]

\subsection{Causality considerations} \label{KK}

It is well known~\cite[e.g.][]{hu1989kramers} that the Fourier transform $\hat{f}(\oo)$ of any causal function $f(t)$ satisfies the Kramers-Kronig relations
\begin{equation} \label{visco10cm}
  \Re[\hat{f}(\omega)] = -\Hcal[\Im[\hat{f}]](\oo) \quad\text{and}\quad
  \Im[\hat{f}(\omega)] = \Hcal[\Re[\hat{f}]](\oo), \qquad\text{i.e.}\quad
   \hat{f}(\oo) = \ii\hh\Hcal[\hat{f}](\oo) ,
\end{equation}
where $\Hcal[g]$ is the Hilbert transform~\cite{king:vol12} of a function $g$ given by
\begin{equation}
  \Hcal[g](\oo)
 := \frac{1}{\pi} \Big( \text{PV}\frac{1}{\oo} \Big)\star g(\oo)
 = \frac{1}{\pi} \int_{\mathbb{R}} \frac{g(\vartheta)}{\omega-\vartheta} \dth, \label{H:def}
\end{equation}
interpreted in the Cauchy principal value sense (see Appendix~\ref{A4} for details). Conversely, satisfaction of conditions~\eqref{visco10cm} implies the causality of~$f(t)$~\cite{schwartz:62}, so that the causal character of $C(t),J(t)$ and $E(t)$ can be ascertained by verifying respectively that $\hat{C}(\omega),\hat{J}(\omega)$ and $\hat{E}(\omega)$ obey~\eqref{visco10cm}.

\begin{prop}\label{KK:prop}
Every complex modulus function~$\hat{C}(\omega)$ in~\eqref{poles1}, that derives from the phenomenological ODE~\eqref{scalarDE}, is the Fourier image of a \emph{causal memory function} $C(t)$.     
\end{prop} 

\begin{proof}
By virtue of their properties outlined in Section~\ref{models:EDO}, complex modulus functions in~\eqref{poles1} admit partial fraction expansions of the form
\begin{equation}
  \hat{C}(\omega)
  = \sum_{i=1}^{M} \sum_{j=1}^{m_i} \frac{\alpha_{ij}}{(\oo-\oo_i)^j} + \Rcal(-\ii\oo), 
\label{aux06}
\end{equation}
where $\alpha_{ij}\!\in\!\mathbb{C}$ are constants; $\oo_i \!=\!-\ii|\oo_i|$ are purely imaginary zeros of the polynomial $\omega\mapsto\Pcal(\ii\hh\oo)$ with respective multiplicities $m_i$, and $\Rcal(-\ii\hh\oo)$ is a polynomial of degree $N\!-\!M$ (with $\Rcal=0$ if $N<M$). Since the Hilbert transform and derivative operators are known to commute~\cite[Section~4.8]{king:vol12}, we demonstrate the claim using the property
\begin{equation}
  \Hcal\Big[\frac{1}{(\oo-\oo_i)^{j+1}}\Big]
 = \frac{(-1)^j}{j!} \Hcal\Big[\frac{\text{d}^j}{\text{d}\oo^j}\Big(\frac{1}{\oo-\oo_i}\Big)\Big]
 = \frac{(-1)^j}{j!} \frac{\text{d}^j}{\text{d}\oo^j} \Hcal\Big[\frac{1}{\oo-\oo_i}\Big], \qquad j\!=\! 0,1,2 \ldots
\end{equation}
and the known Hilbert transform pair~\cite[Vol.~2, Appendix~1, Table 1.2, eq.~(2.2)]{king:vol12}
\begin{equation}
  \Hcal\Big[\frac{1}{\oo-\oo_i}\Big] = \frac{-\ii}{\oo-\oo_i}, \label{aux03}
\end{equation}
inferring that the proper rational fraction part of~\eqref{aux06} satisfies~\eqref{visco10cm} and is thus~\cite{schwartz:62} the Fourier image of a causal distribution. Further, the polynomial part $\Rcal(-\ii\oo)$ of~\eqref{aux06} is a weighted sum of powers $(-\ii\omega)^k$, $k\!=\!\overline{0,N\!-\!M}$ that verify $(-\ii\oo)^k=\Fcal[\delta^{(k)}(t)]$ (where $\delta^{(k)} \!:=\!\dd^k \delta/\dd t^k$ in the sense of distributions) and is thus also the Fourier image of a causal distribution.    
\end{proof}

\subsection{Plane-stress relaxation and creep tensors}

On the basis of~\eqref{visco1} and~\eqref{visco1cm}, the plane-stress memory and relaxation tensors are related through 
\begin{equation}\label{aux1}
\bfC^{\sigma\!}(t) \star\bfeps\psup(\bfx,t) \,=\, \bfE^{\sigma\!}(t) \star\dot{\bfeps}\psup(\bfx,t). 
\end{equation}
Letting $\bfeps\psup(\bfx,t) \!=\! H(t) \hh \bfeps\psup(\bfx)$ due to Fig.~\ref{relax1} and recalling the Fourier transforms of~$H(t)$ and~$\delta(t)$ given in Table~\ref{tab:fourier}, application of the convolution theorem to~$\mathcal{F}[\eqref{aux1}]$ yields 
\begin{equation} \label{aux2}
\hCS(\omega)\dip \bfeps\psup(\bfx)\, \Big(\frac{\ii}{\omega} + \pi\delta(\omega) \Big) \,=\, \hCE(\omega)\dip \bfeps\psup(\bfx). 
\end{equation}
By the arbitrariness of (kinematically-admissible) $\bfeps\psup(\bfx)$, we obtain the relationship
\begin{equation} \label{E:ps}
\hCE(\omega) = \frac{\ii}{\omega}\hh \hCS(\omega) + \pi\hh \delta(\omega) \hh \hCS(0).
\end{equation}

In an analogous way, \eqref{visco1cm} and \eqref{visco2} yield $\dot{\bfsig}{}\psup =\dot{\bfC}{}^{\sigma}\!\star \bfeps\psup = (\dot{\bfC}{}^{\sigma}\!\star\bfCJ)\star\dot{\bfsig}{}\psup$ and, under Fourier transform,
\begin{equation}
  -\ii\omega\hh \hCS(\omega)\dip\hCJ(\omega) = \IS.
\end{equation}
When $\hCS(0)\not=\bfze$, the solution $\hCJ(\omega)$ of the above equation that defines a causal plane-stress creep tensor $\bfCJ(t)$ is given by
\begin{equation}
  \hCJ(\omega) = \frac{\ii}{\omega}\big(\hCS(\omega)\big)^\dagger + \pi\delta(\omega)\big(\hCS(0)\big)^\dagger
\label{J:ps}
\end{equation}
similar to the case (i) discussed in Section~\ref{relcreep}. We note that the Dirac delta terms in~\eqref{E:ps} and~\eqref{J:ps} represent the instantaneous elastic response of a viscoelastic solid under the plane stress condition. 

When $\hCS(0)=\bfze$, on the other hand, with the zero of $\hCS(\omega)$ at $\omega=0$ assumed to be of single multiplicity, we proceed similarly as in case (ii) of Section~\ref{relcreep}. Specifically, setting $\hCS(\omega)=\omega\hBS(\omega)$ (where $\hBS(\omega)$ is invertible for any real $\omega$), $\hCJ(\omega)$ is sought as the (Fourier image of the) causal solution of $-\ii\omega^2\hBS(\omega)\dip\hCJ(\omega) = \IS$. This results in the expression 
\begin{equation}
  \hCJ(\omega)
 = \frac{\ii}{\omega}\big(\hCS(\omega)\big)^\dagger
 - \pi\delta'(\omega) \big(\hBS(0)\big)^\dagger
 + \pi\delta(\omega)\bigg(\der{\hBS}{\omega}(0)\bigg)^\dagger.
\label{J:ps2}
\end{equation}

\paragraph{Evaluation of~$\bfJ^\sigma$ and~$\bfE^\sigma$ via the Laplace transform}

In the literature, the relaxation and creep functions of a linear viscoelastic model are often interrelated \cite{findley2013creep} by way of the Laplace transform pair
\begin{equation} \label{laplace}
\check{f}(s) = \Lcal[f](s) = \int_{0}^{\infty} f(t) \hh e^{-s\hh t}\hh \dt, \qquad\quad
f(t) = \Lcal^{-1}[\check{f]}(s) = \frac{1}{2\pi\ii} \lim_{\ell\to\infty} \int_{\gamma-\ii\hh\ell}^{\gamma+\ii\hh\ell} \check{f}(s) \hh e^{s\hh t} \ds,
\end{equation}
where $s\!\in\!\mathbb{C}$ and $\gamma\!>\!0$ is such all singularities of $\check{f}(s)$ in the complex $s$-plane are located ``to the left'' of the line of integration. In situations where $\,\,\check{\!\!\bfJ}(s)$ is prescribed, we have the obvious relationship
\begin{equation} \label{visco4x}
\bfJ^{\hh\sigma}\!(t) =\, \Lcal^{-1}\big[\, \bfP\dip\,\,\check{\!\!\bfJ}(s)\dip\bfP \,\big]. 
\end{equation}
On the other hand, application of the Laplace transform to~\eqref{visco1} followed by manipulation \`a la~\eqref{stress3} in the transformed domain yields
\begin{align}
  \bfsig\psup(\bfxi,t) &= \bfE^{\hh\sigma}(t)\star \dot{\bfeps}\psup(\bfxi,t), \\
  \bfE^{\hh\sigma}(t)
 &= \Lcal^{-1}\big[\, (\bfP \dip \,\,\check{\!\!\bfE}^{-1}\!(s) \dip \bfP)^{\dagger} \,\big]
 = \frac{1}{2\pi\ii} \lim_{\ell\to\infty} \int_{\gamma-\ii\hh\ell}^{\gamma+\ii\hh\ell} s^{-2}(\bfP \dip \,\,\check{\!\!\bfJ}(s) \dip \bfP)^{\dagger} \hh e^{s\hh t} \ds, \label{visco7}
\end{align}
where (i) $\Re(s)\!>\!0$ is assumed to be sufficiently large so that the integrals entailed in computing the components of $\bfE^{\hh\sigma}$ are convergent, and (ii) the second formula for~$\bfE^{\hh\sigma}$ is due to the relationship
\begin{equation} \label{visco5}
\check{\bfeps}\psup(\bfx,s) = s\hh \check{\bfJ^{\sigma}}(s)\dip\check{\bfsig}\psup(\bfx,s)  \quad \Longrightarrow \quad \check{\bfsig}\psup(\bfx,s) = s^{-1}\hh (\check{\bfJ^{\sigma}}(s))^{\dagger}\dip\check{\bfeps}\psup(\bfx,s).
\end{equation}

\section{Thermodynamic potentials describing linear viscoelastic behavior under the plane stress condition} \label{GSM}

\noindent To illustrate the effect of the plane stress environment on thermodynamic description of linear viscoelastic solids, we consider the framework of generalized standard materials (GSM) \cite{ger:nqs:suq:83,Halp75,maugin:92}. The latter describes a broad family of constitutive models for dissipative solids that are formulated in terms of two convex thermodynamic potentials, namely the Helmholtz free energy $\psi$ and dissipation potential $\varphi$. 

\subsection{Generalized standard materials}

On denoting by $\bfeps(\bfx,t)$ the linearized strain tensor (as before) and by~$\bfal(\bfx,t)$ the tensor of internal thermodynamic variables, within this framework we write $\psi=\psi(\bfeps,\bfal)$ and $\varphi=\varphi(\bfepsd,\bfald)$. By decomposing the Cauchy stress tensor $\bfsig$ into a reversible (``elastic'') component $\bfsige$ and irreversible (``viscous'')  component $\bfsigv$ as 
\begin{equation} \label{stress0}
\bfsig(\bfx,t) = \bfsige(\bfx,t) +\bfsigv(\bfx,t),   
\end{equation}
the GSM constitutive relationships yield $\bfsige$, $\bfsigv$ and the tensor of thermodynamic tensions $\bfA$ (energetically conjugate to~$\bfal$) as
\begin{equation} \label{stress}
\bfsige = \del{\eps}\psi, \qquad\quad \bfsigv=\del{\epsd}\varphi, \qquad\quad
\bfA = -\del{\alpha}\psi = \del{\alphad}\varphi. 
\end{equation}
Assuming isothermal conditions, the power per unit volume $\mathfrak{D}=\mathfrak{T\dot{S}}$ ($\mathfrak{T}=\,$temperature, $\mathfrak{S}=\,$specific entropy) dissipated by a material obeying~\eqref{stress} reads
\begin{equation} \notag
\mathfrak{D} = \bfsigv\dip\bfepsd + \bfA\dip\bfald,  \label{dissip}
\end{equation}
which exposes~$\bfsigv$ and~$\bfA$ as the generators of entropy production. 

For the present case of linear viscoelasticity, both potentials $\psi$ and $\varphi$ are quadratic functions of their arguments \cite{jmps2024}; Specifically, one has 
\begin{equation} \label{phi:psi:def}
\begin{aligned}
 \psi(\bfeps,\bfal)
 &= \demi\bigl( \bfeps\dip\Ce\dip\bfeps + 2\bfeps\dip\CmT\dip\bfal + \bfal\dip\Ca\dip\bfal \bigr), \\*[2mm]
 \varphi(\bfepsd,\bfald)
 &= \demi\bigl( \bfepsd\dip\De\dip\bfepsd + 2\bfepsd\dip\DmT\dip\bfald + \bfald\dip\Da\dip\bfald \bigr),
\end{aligned} 
\end{equation}
where (i) the second-order symmetric tensor $\bfal$ (i.e.~the ``viscous strain'' tensor) collects internal variables of the model; (ii)  $\Ce,\Ca,\De$ and~$\Da$ are fourth-order tensors, characterized by major and minor symmetries, that define non-negative quadratic forms over the second-order symmetric tensors; and (iii)  $\Cm$ and~$\Dm$ are the fourth-order tensors carrying minor symmetries and must be such that $\psi(\bfeps,\bfal)$ and $\varphi(\bfepsd,\bfald)$ are positive -- and so convex -- functions. Motivated by physical considerations, we further assume that $\Ce$ and $\Da$ are invertible (and so positive definite), and that $\Cm$ and~$\Dm$ also carry the major symmetry. 

From the last equality in~\eqref{stress} and~\eqref{phi:psi:def}, $\bfal$ can be solved for locally in terms of~$\bfeps$ as 
\begin{equation}
\bfal(\bfx,t) = -\FS[\Cm\dip\bfeps\shp\Dm\dip\bfepsd](\bfx,t), \label{alpha:expr}
\end{equation}
where the tensor-valued convolution map $\bfs\mapsto\FS[\bfs]$ is given by
\begin{equation} \label{alpheps}
\FS[\bfs](\bfx,t) = \int_{0}^{t} \exp[-\QS(t\shm\tau)]\dip\Da^{-1}\dip\bfs(\bfx,\tau) \dtau, \qquad
\QS\shdeq\Da^{-1}\dip\Ca. 
\end{equation}
By virtue of~\eqref{stress0}, \eqref{stress} and~\eqref{alpheps}, the hereditary stress-strain relationship due to~\eqref{phi:psi:def} is obtained as 
\begin{equation} \label{gsm:sigeps}
\bfsig(\bfx,t)
= \CS\Isub \dip\bfeps(\bfx,t) + \DS\Isub\dip\bfepsd(\bfx,t) - \CHT\!\dip\FS[\CH\dip\bfeps](\bfx,t), 
\end{equation}
where 
\begin{equation}
\CH:=\Cm-\Ca\dip\Da^{-1}\dip\Dm \label{alpha(t)} \label{alpha(t)}
\end{equation}
is the fourth-order tensor featured in the convolution term, and  
\begin{equation}
\CS\Isub = \Ce - \CmT\dip\Ca^{-1}\dip\Cm + \CHT\dip\Ca^{-1}\dip\CH, \qquad~ \DS\Isub = \De-\DmT\dip\Da^{-1}\dip\Dm
\label{C:inst}
\end{equation}
denote respectively the \emph{instantaneous} elasticity and viscosity tensors endowed with both major and minor symmetries \cite{jmps2024}. On rewriting~\eqref{gsm:sigeps} in the format given by~\eqref{visco1cm}, we obtain a general expression for the GSM tensor memory function as 
\begin{equation} \label{gsm:memory}
\bfC(t)
\,=\, \CS\Isub\hh \delta(t) \hh+\hh \DS\Isub\hh \delta'(t) \hh-\hh \CHT\!\dip\exp[-\QS\hh t]\dip\Da^{-1}, 
\end{equation}    
invoking the previously encountered distributions~$\delta(\dotp)$, $\delta'(\dotp)$ (see Table~\ref{tab1}) and a set of characteristic times featured by the exponent tensor~$\QS$. Before proceeding to the next section, we recall the distinction between characteristic and relaxation times made in Remark~\ref{char-relax}. 

\subsection{Limitation of the GSM description} \label{GSMltd}

On invoking the series expansion of a tensor exponential function and rearranging the terms, we observe that 
\begin{equation}
\exp[-\QS\hh t]\dip\Da^{-1} = \Da^{-1/2}\exp[-\QS'\hh t]\dip\Da^{-1/2}, \qquad 
\QS':= \Da^{1/2}\dip\QS\dip\Da^{-1/2} = \Da^{-1/2}\dip\Ca\dip\Da^{-1/2}
\end{equation}    
where~$\QS'$ is symmetric, which demonstrates that the GSM model~\eqref{gsm:memory} features 21 independent characteristic times. By virtue of the Voigt representation~\cite{voigt1910lehrbuch} of~$\QS'$ in terms of a $6\times 6$ matrix~$\bfQ'$ (see Appendix~\ref{A:Voigt}), one can show via the series expansion of~$\exp[-\bfQ\hh t]$ that these 21 \emph{characteristic times} generate up to six distinct \emph{relaxation times}, corresponding to the six eigenvalues of~$\bfQ'$. Specifically, by writing 
\[
\bfQ' = \bfU\!\bfLa\hh \bfU\Tsup, \qquad \bfU\bfU\Tsup = \bfU\Tsup\bfU = \bfI, \qquad 
\bfLa = \text{diag}[\tau_j^{-1},\; j\!=\!\overline{1,6}]
\]
where~$\tau_j\!>\!0$ and $\bfU\in\Rbb^{6\times 6}$ is an orthogonal matrix, we obtain 
\begin{multline}
  \exp[-\bfQ't] \,=\, \bfI + \sum_{n=1}^{\infty}\inv{n!} (-\bfQ't)^n
  \,=\, \bfU\bfU\Tsup + \sum_{n=1}^{\infty} \inv{n!} (-t\, \bfU\!\bfLa\hh \bfU\Tsup)^n
  \,=\, \bfU\bfU\Tsup + \sum_{n=1}^{\infty} \inv{n!}\, \bfU(-t\hh \bfLa)^n \, \bfU\Tsup \\
  \,=\, \bfU \Big\{ \bfI + \sum_{n=1}^{\infty} \inv{n!} (-t\bfLa)^n \Big\} \, \bfU\Tsup
  \,=\, \bfU \Big\{ \sum_{n=0}^{\infty} \inv{n!} \, \text{diag}[(-t/\tau_j)^n,\; j\!=\!\overline{1,6}] \Big\}\, \bfU\Tsup
  \,=\, \bfU \text{diag}[e^{-t/\tau_j},\; j\!=\!\overline{1,6}] \, \bfU\Tsup. 
\end{multline}
From this result, we find that the GSM model is not capable of describing phenomenological viscoelastic solids, e.g. those rooted in~\eqref{scalarDE}, that feature more than (i) six relaxation times for general anisotropic solids, or (ii) two relaxation times for isotropic materials, since $\bfLa$ then carries only two distinct eigenvalues.

\subsection{Plane-stress expressions for the GSM potentials}

Consistent with the analysis in Section~\ref{psproj}, we assume that all flux fields in~\eqref{stress} adhere to the plane stress condition. Specifically we postulate 
\begin{equation}\label{stress2x}
\bfsige = \bfP\dip\bfsige \hh:=\hh \bfsig\psupe, \qquad
\bfsigv = \bfP\dip\bfsigv \hh:=\hh \bfsig\psupv, \qquad
\bfA = \bfP\dip\bfA \hh:=\hh \bfA\psup,
\end{equation}
while keeping the kinematic fields unabridged, i.e. 
\begin{equation}\label{strain2x}
\bfeps = \bfP\dip\bfeps + \bfR\dip\bfeps \,:=\hh \bfeps\psup + \bfeps\rsup, \qquad
\bfal = \bfP\dip\bfal + \bfR\dip\bfal \,:=\hh \bfal\psup + \bfal\rsup. 
\end{equation}
From~\eqref{stress} and~\eqref{phi:psi:def}, on the other hand, we find the fluxes~\eqref{stress2x} and kinematic quantities~\eqref{strain2x} to be related as 
\begin{equation} \label{auxx1}
\begin{aligned}
& \text{(a)} \quad ~~\bfsige =&\hspace*{-2mm}& \Ce\dip\bfeps + \CmT\dip\bfal, &\qquad&
& \text{(b)} ~~~~\bfsigv =&\hspace*{-2mm}& \De\dip\bfepsd + \DmT\dip\bfald, \\
& \text{(c)} \quad -\bfA     =&\hspace*{-2mm}& \CmT\dip\bfeps + \Ca\dip\bfal, &\qquad&
& \text{(d)} \quad ~~\bfA    =&\hspace*{-2mm}& \DmT\dip\bfepsd + \Da\dip\bfald.
\end{aligned}
\end{equation}

To eliminate the out-of-plane kinematic fields~$\bfeps\rsup$ and~$\bfal\rsup$ from consideration, we first solve~\eqref{auxx1}(a,c) for~$\bfeps$ and~$\bfal$ and make use of~\eqref{stress2x}--\eqref{strain2x}. This yields 
\begin{equation} \label{ps-gsm-compl1}
\begin{aligned}
\bfeps\psup &= (\bfP\dip\Cs\dip\bfP)\dip \bfsig\psupe + (\bfP\dip\Cms\dip\bfP)\dip\bfA\psup, \\
-\bfal\psup &= (\bfP\dip\Cms\dip\bfP)\Tsup\dip\bfsig\psupe + (\bfP\dip\CA\dip\bfP)\dip\bfA\psup, 
\end{aligned}
\end{equation}
where $M_{ijkl}\Tsup\!=\!M_{klij}$ for a generic fourth-order tensor~$\bfM$, and  
\begin{equation} 
\Cs = (\CeS)^{-1}, \qquad \CA = (\CaS)^{-1}, \qquad 
\Cms = (\CeS)^{-1}\dip\Cm\dip\Ca^{-1}   
\end{equation}
with 
\begin{equation} \label{schur:def}
\CeS := \Ce - \CmT\dip\Ca^{-1}\dip\Cm, \qquad~~ \CaS := \Ca - \Cm\dip\Ce^{-1}\dip\CmT 
\end{equation}
denoting the respective Schur complements of $\Ce$ and $\Ca$. Here it is important to note that~$\Cms$ is not guaranteed to carry major symmetry, even though~$\Cm$ (by premise) does. By solving the projected relationships~\eqref{ps-gsm-compl1} for the flux fields, we obtain 
\begin{equation} \label{td-pots-ps1}
\begin{aligned}
\bfsig\psupe &\,=\,  \Ceps\dip\bfeps\psup + \CmTps\dip\bfal\psup, \\
-\bfA\psup   &\,=\,  (\CmTps)\Tsup\dip\bfeps\psup + \Caps\dip\bfal\psup, 
\end{aligned}    
\end{equation}
where the respective plane-stress reductions of $\Ce,\Ca$ and~$\CmT$ featured by the first of~\eqref{phi:psi:def} read 
\begin{equation} \label{C:coeff:conj}
\Ceps = (\SSs)^\dagger, \qquad \Caps = (\SSA)^\dagger, \qquad 
\CmTps = (\SSs)^\dagger\dip(\bfP\dip\Cms\dip\bfP)\Tsup \dip (\bfP\dip\CA\dip\bfP)^\dagger, 
\end{equation}
with 
\begin{equation} \label{schur:def2} 
\begin{aligned}
\SSs := \bfP\dip\Cs\dip\bfP - 
(\bfP\dip\Cms\dip\bfP) \dip (\bfP\dip\CA\dip\bfP)^\dagger \dip (\bfP\dip\Cms\dip\bfP)\Tsup, \\
 \SSA := \bfP\dip\CA\dip\bfP - 
(\bfP\dip\Cms\dip\bfP)\Tsup \dip (\bfP\dip\Cs\dip\bfP)^\dagger \dip (\bfP\dip\Cms\dip\bfP), 
\end{aligned}
\end{equation}
denoting the respective Schur complements of $\bfP\dip\Cs\dip\bfP$ and $\bfP\dip\CA\dip\bfP$.

Similarly, starting from~\eqref{auxx1}(b,d) we obtain  the respective plane-stress reductions of $\De,\Da$ and~$\DmT$ featured by the second of~\eqref{phi:psi:def} as
\begin{equation} \label{td-pots-ps2}
\begin{aligned}
\bfsig\psupv &\,=\,  \Deps\dip\bfepsd\psup + \DmTps\dip\bfald\psup, \\
\bfA\psup   &\,=\,  (\DmTps)\Tsup\dip\bfepsd\psup + \Daps\dip\bfald\psup, 
\end{aligned}    
\end{equation}
where 
\begin{equation}
\Deps = (\STs)^\dagger, \qquad \Daps = (\STA)^\dagger, \qquad 
\DmTps = (\STs)^\dagger\dip(\bfP\dip\Dms\dip\bfP)\Tsup \dip (\bfP\dip\DA\dip\bfP)^\dagger 
\end{equation}
with
\begin{equation} 
\begin{aligned}
\STs := \bfP\dip\Ds\dip\bfP - 
(\bfP\dip\Dms\dip\bfP) \dip (\bfP\dip\DA\dip\bfP)^\dagger \dip (\bfP\dip\Dms\dip\bfP)\Tsup, \\
 \STA := \bfP\dip\DA\dip\bfP - 
(\bfP\dip\Dms\dip\bfP)\Tsup \dip (\bfP\dip\Ds\dip\bfP)^\dagger \dip (\bfP\dip\Dms\dip\bfP), 
\end{aligned}
\end{equation}
and 
\begin{equation}
 \Ds = (\DeS)^{-1}, \qquad \DA = (\DaS)^{-1}, \qquad 
 \Dms = (\DeS)^{-1}\dip\Dm\dip\Da^{-1} \label{D:coeff:conj}
\end{equation}
with
\begin{equation}
\DeS := \De - \DmT\dip\Da^{-1}\dip\Dm, \qquad~~ \DaS := \Da - \Dm\dip\De^{-1}\dip\DmT. \label{schur:defD}
\end{equation}

With~\eqref{td-pots-ps1} and~\eqref{td-pots-ps2} in place, thermodynamic potentials~\eqref{phi:psi:def} describing  linear viscoelastic solid under the plane stress condition become 
\begin{equation} \label{phi:psi:def-pstress1}
\begin{aligned}
\psi &\,=\, \demi\bigl(\bfeps\psup\dip\Ceps\dip\bfeps\psup + 2\bfeps\psup\dip\CmTps\dip\bfal\psup +
\bfal\psup\dip\Caps\dip\bfal\psup \bigr), \\*[2mm]
\varphi &\,=\, \demi\bigl(\bfepsd\dip\Deps\dip\bfepsd\psup + 2\bfepsd\psup\dip\DmTps\dip\bfald\psup + \bfald\psup\dip\Daps\dip\bfald\psup \bigr).
\end{aligned} 
\end{equation}
By retaining the format of the original 3D model~\eqref{phi:psi:def} (with all material tensors replaced by their plane-stress counterparts), we find that general properties such as~\eqref{alpheps}, \eqref{gsm:sigeps} and~\eqref{gsm:memory} carry over to the plane-stress context, provided that the tensor inverse $(\bdot)^{-1}$ is superseded by its Moore-Penrose counterpart~$(\bdot)^\dagger$. On revisiting the discussion in Section~\ref{GSMltd}, we also observe that the plane-stress constitutive model~\eqref{phi:psi:def-pstress1} permits \emph{at most three} relaxation times.  

\begin{remark}
As an important practical consequence, the format similarity between~\eqref{phi:psi:def} and~\eqref{phi:psi:def-pstress1} demonstrates that the error-in-constitutive-relation framework~\cite{jmps2024} can be readily deployed toward material identification of sheet-like (dissipative) solid specimens subjected to the plane stress condition.    
\end{remark}

In the case of non-dissipative material behavior where $\CmT=\Ca=\De=\DmT=\Da=\bfze$, it is easy to show that the plane-stress reduction of~\eqref{phi:psi:def} reads 
\begin{equation} \label{phi:psi:def-pstress2}
 \psi \,=\, \demi\bfeps\psup\dip\Ceps\dip\bfeps\psup, \qquad 
 \Ceps = (\bfP\dip\Ce^{-1}\dip\bfP)^\dagger, \qquad \varphi=0
 \end{equation}
which is consistent with~\eqref{stress3}--\eqref{stress4}.

\section{Examples of the plane-stress memory functions in bulk}
\label{A3}

\noindent By the second of~\eqref{isovis8} the complex shear modulus~$\hmuCS(\omega)$ of an isotropic viscoelastic solid, and consequently its memory function in shear~$\muCS(t)$, are unaffected by the plane stress condition. This is, however, not the case with the complex bulk modulus~$\hkaCS(\omega)$ due to the first of~\eqref{isovis8}. Accordingly, it is of interest to compute (via inverse Fourier transform) the plane-stress memory function in bulk, $\kaCS(t)$, for a set of classical viscoelastic models including those in Fig.~\ref{viscomod1}. In what follows, (i) the relevant characteristic time intervals are assumed for convenience to be non-oriented in that $[\tau_1,\tau_2]:=[\min\{\tau_1,\tau_2\},\max\{\tau_1,\tau_2\}]$, and (ii) our computations rely on the set of Fourier transform pairs gathered in Table~\ref{tab:fourier} (see Appendix~\ref{A:FTP} for details).

\begin{table}[h]
\centering \renewcommand{\arraystretch}{1.35}
\begin{tabular}{|c||c|c|}\hline
 (a) & $1$ & $2\pi\delta(\omega)$ \\ & $\delta(t)$ & 1 \\ \hline
 (b) & $t$ & $-2\ii\pi\delta'(\omega)$ \\ & $\delta'(t)$ & $-\ii\omega$ \\ \hline
\end{tabular} \qquad\qquad
\begin{tabular}{|c||c|c|}\hline
 (c) & $\text{sgn}(t)$ & $2\ii/\omega$ \\ \hline
 (d) & $t\, \text{sgn}(t)$ & $-2/\omega^2$ \\ \hline
 (e) & $H(t)$ & $\pi \delta(\omega) + \ii/\omega$ \\ \hline
 (f) & $H(t) \, e^{-t/\tau}$ & $\tau/(1-\ii\tau\hh\omega)$ \\ \hline
\end{tabular}
\caption{Relevant Fourier transform pairs involving distributions.}
\label{tab:fourier}
\end{table}

\subsection{Kelvin-Voigt solid} \label{kelvin}

Letting $\gamma\!=\!G$ and $\tau\!=\!\eta/G$ according to~(\ref{fourx11}a), for the viscoelastic model described in Fig.~\ref{viscomod1}(a) one has
\begin{equation} \label{four7}
\sigma(t) = \gamma \hh \eps(t) + \eta\hh \dot{\eps}(t) ~~ \Rightarrow ~~
  \hat{C}(\omega) = \gamma\big( 1 - \ii\tau\omega \big) ~~\Rightarrow~~
  C(t) = \gamma\big[\tau\hh \delta'(t) + \delta(t)\big]
\end{equation}
by virtue of the Fourier transform pairs~(a) and~(b) in Table~\ref{tab:fourier}. This allows us to consider an isotropic Kelvin-Voigt solid, whose memory functions $\kappa\Csub(t)$ and $\mu\Csub(t)$ obey the above model with the respective elastic constants $\kappa$ and~$\gamma$ and characteristic times $\tau$ and~$\theta$. In this case the plane-stress complex bulk modulus $\hkaCS(\omega)$,
given by~\eqref{isovis8} with $\hkaC(\omega) = \kappa(1-\ii\tau\omega)$ and $\hmuC(\omega) = \gamma(1-\ii\theta\omega)$, can be computed as 
\begin{equation} \label{ksig:comp}
\hkaCS(\omega)
= \kaS\frac{(1-\ii\tau\omega )(1-\ii\theta\omega )}{1-\ii\tau^{\sigma}\omega}
  = \kaS \Big( 1-\ii\frac{\tau\theta}{\tau^{\sigma}}\omega -\alpha + \frac{\alpha}{1-\ii\tau^{\sigma}\omega} \Big), 
\end{equation}
featuring the (plane-stress) elastic bulk modulus~$\kaS$ and relaxation time~$\tau^{\sigma}$ given by 
\begin{equation} \label{ksig:elas}
\kaS := \frac{9\kappa \gamma}{3\kappa+4\gamma}, \qquad 
\tau^{\sigma} := \frac{3\kappa\hh \tau + 4\gamma\hh \theta}{3\kappa+4\gamma} 
\end{equation}
and auxiliary constant 
\begin{equation}
\alpha = \frac{(\tau^{\sigma}-\tau)(\tau^{\sigma}-\theta)}{(\tau^{\sigma})^2}
 = -\frac{12\kappa\mu(\tau-\theta)^2}{(3\kappa\hh\tau + 4\gamma\hh\theta)^2}. \label{aux01}
\end{equation}
As a result, the plane-stress memory function in bulk $\kaCS(t)$ is obtained (via the Fourier pairs~(a), (b) and~(f) in Table~\ref{tab:fourier}) as the causal distribution 
\begin{equation} \label{visco11cm}
\kaCS(t) = \kaS \Big[ \frac{\tau\theta}{\tau^{\sigma}}\hh\delta'(t) + \alpha\hh\delta(t) + \frac{1-\alpha}{\tau^{\sigma}} H(t) \, e^{-t/\tau^{\sigma}} \Big]. 
\end{equation}

From~\eqref{aux01}, we observe that $(\tau^{\sigma}-\tau)(\tau^{\sigma}-\theta)<0$, showing that $\tau^{\sigma}\in[\tau,\theta]$. It is also interesting to note that~\eqref{visco11cm} resembles the memory function of the Jeffreys model (dashpot~$\eta_1$ connected in series with a Kelvin-Voigt model (dashpot~$\eta_2$, spring~$\gamma$)) whose memory function reads
\begin{equation} \label{visco12cm}
C(t) = \frac{\eta_1}{\tau_1} \Big[\tau_2\hh\delta'(t) + \frac{\tau_1-\tau_2}{\tau_1}\hh\delta(t) -
\frac{\tau_1-\tau_2}{\tau_1^2}\hh H(t)\,e^{-t/\tau_1} \Big], \qquad \tau_1=\frac{\eta_1\!+\eta_2}{\gamma}, \quad \tau_2 = \frac{\eta_2}{\gamma}.
\end{equation}

\subsection{Maxwell material} 

With reference to Fig.~\ref{viscomod1}(b), the Maxwell viscoelastic model satisfies
\[
\tau\dot\sigma(t) + \sigma(t) = \eta\dot\eps(t)  ~~\Rightarrow~~
\hat{C}(\omega) = \gamma\hh \frac{\tau\omega}{\ii+\tau\omega} 
~~\Rightarrow~~
C(t) = \gamma\Big[\delta(t) - \frac{1}{\tau} \hh H(t) \hh e^{-t/\tau}\Big],
\]
where $\gamma\!=\!G$ and $\tau\!=\!\eta/G$ due to~(\ref{fourx11}b), and $C(t)$ is obtained via the Fourier pairs~(a) and~(f) in Table~\ref{tab:fourier}.
Let us now consider an isotropic Maxwell material whose memory functions $\kappa\Csub$ and $\mu\Csub$ obey the above model with respective elastic constants $\kappa$ and~$\gamma$ and relaxation times $\tau$ and~$\theta$. The plane-stress complex bulk modulus $\hkaCS(\omega)$, given by~\eqref{isovis8} with $\hkaC(\omega) = \kappa \hh \frac{\tau\hh \omega}{\ii+\tau\hh\omega}$ and $\hmuC(\omega) = \gamma \hh \frac{\theta\hh \omega}{\ii+\theta\hh\omega}$, is obtained after some algebra as
\begin{equation} \label{kapas-maxw}
\hkaCS(\omega) = \kaS\frac{\tau^{\sigma}\omega}{\ii+\tau^{\sigma}\omega}, 
\end{equation}
featuring the (plane-stress) elastic bulk modulus $\kaS$ due to~\eqref{ksig:elas} and relaxation time 
\begin{equation} \label{taus-maxw}
\tau^{\sigma} = \frac{(3\kappa+4\gamma)\tau\hh \theta}{3\kappa\hh\tau + 4\gamma\hh\theta}.     
\end{equation}
Again, it is easily shown that $\tau^{\sigma}\in[\tau,\theta]$. Consequently, the plane-stress memory function in bulk $\kaCS(t)$ is found to obey the Maxwell model with elastic constant $\kaS$ and relaxation time $\tau^{\sigma}$, namely 
\begin{equation} \notag
\kaCS(t) = \kaS \Big[\delta(t) - \frac{1}{\tau^{\sigma}} \hh H(t) \hh e^{-t/\tau^{\sigma}}\Big]. 
\end{equation}

\subsection{Standard linear solid} \label{secSLS}

For the standard linear solid (SLS) model shown in Fig.~\ref{viscomod1}(c), letting $\gamma,\tau_1$ and~$\tau_2$ take values according to~(\ref{fourx11}c) we obtain
\begin{multline} \label{four12}
\sigma(t) + \tau_1\dot\sigma(t)  = \gamma\big(\eps(t) + \tau_2\hh\dot\eps(t)\big), \quad \tau_1\!<\!\tau_2 ~~\Rightarrow~~ \\ 
\hat{C}(\omega) = \gamma \hh \frac{1-\ii\tau_2\hh\omega}{1-\ii\tau_1\hh\omega}
~~~\Rightarrow~~
  C(t) = \gamma\Big[\frac{\tau_2}{\tau_1} \delta(t) - \frac{\tau_2-\tau_1}{\tau_1^2} \hh H(t) e^{-t/\tau_1} \Big].
\end{multline}
In this setting, we consider an isotropic solid endowed with the memory function $\kappa\Csub$ (resp.~$\mu\Csub$) due to the SLS model featuring the elastic constant $\kappa$ (resp.~$\gamma$) and characteristic times $\tau_1\!<\!\tau_2$ (resp.~$\theta_1\!<\!\theta_2$). After some algebra, the plane-stress complex bulk modulus $\hkaCS(\omega)$, given by~\eqref{isovis8} with $\hkaC(\omega) = \kappa \hh \frac{1-\ii\tau_2\omega}{1-\ii\tau_1\omega}$ and $\hmuC(\omega) = \gamma \hh \frac{1-\ii\theta_2\omega}{1-\ii\theta_1\omega}$, can be shown to read
\begin{equation} \label{cmodsls}
\hkaCS(\omega)
\,=\, \kaS\,\frac{(1-\ii\tau_2\omega)(1-\ii\theta_2\omega)}{A(\ii\omega)^2-B(\ii\omega)+1}
\,=\, \kaS \Big( \frac{\tau_2\theta_2}{A}
 + \frac{\alpha_1}{1-\ii\tau^{\sigma}_1\omega} + \frac{\alpha_2}{1-\ii\tau^{\sigma}_2\omega} \Big), 
\end{equation}
where the (plane-stress) elastic bulk modulus $\kaS$ is again given by~\eqref{ksig:elas}; 
\begin{equation} \label{aux05}
A = \frac{3\kappa\hh\tau_2\hh\theta_1\!+4\gamma\hh\tau_1\hh\theta_2}{3\kappa+4\gamma}>0, \quad
B = \frac{3\kappa(\tau_2+\theta_1)+4\gamma(\tau_1\!+\theta_2)}{3\kappa+4\gamma}>0;
\end{equation}
and
\begin{equation} \label{aux05b}
\alpha_j
 = \frac{(\tau^{\sigma}_j-\tau_2)(\tau^{\sigma}_j-\theta_2)}{\tau^{\sigma}_j(\tau^{\sigma}_j-\tau^{\sigma}_{3-j})}, ~~~j=\overline{1,2}.
\end{equation}
In~\eqref{aux05}--\eqref{aux05b}, the plane-stress relaxation times $\tau^{\sigma}_1<\tau^{\sigma}_2$ are given by $\tau^{\sigma}_j = 1/z_j$, with  $z_1>z_2>0$ being the  roots of
\begin{equation} \label{D=0}
\mathcal{D}(z) := Az^2 - Bz + 1 = 0. 
\end{equation}
For completeness, we note that $z_{1/2}$ are necessarily (i) real-valued due to the fact that $\tau_1\!<\!\tau_2$ and $\theta_1\!<\!\theta_2$ which yields 
\begin{equation}
B^2-4A \,=\, \frac{\big(3\kappa(\tau_2-\theta_1) + 4\gamma(\tau_1-\theta_2) \big)^2 + 48\kappa\gamma(\tau_1-\tau_2)(\theta_1-\theta_2)}{(3\kappa+4\gamma)^2} > 0, \label{AB:discr}
\end{equation}
and (ii) positive thanks to the positivity of~$A$ and~$B$. As a result, the plane-stress memory function in bulk can be computed as
\begin{equation} \label{memory-sls}
\kaCS(t) = \kaS \Big[ \frac{\tau_2\theta_2}{A} \delta(t)
+ \frac{\alpha_1}{\tau^{\sigma}_1} H(t) \, e^{-t/\tau^{\sigma}_1}
+ \frac{\alpha_2}{\tau^{\sigma}_2} H(t) \, e^{-t/\tau^{\sigma}_2} \Big].
\end{equation}
To assess the nature of the featured relaxation times~$\tau^{\sigma}_{j}$ $(j\!=\!\overline{1,2})$, we note that 
\[
\mathcal{D}(\tau^{-1}_1) = \frac{3\kappa(\theta_1-\tau_1)(\tau_2-\tau_1)}{(3\kappa+4\gamma)\tau_1^2}, \qquad
\mathcal{D}(\theta^{-1}_1) = -\frac{4\gamma(\theta_1-\tau_1)(\theta_2-\theta_1)}{(3\kappa+4\gamma)\theta_1^2}.
\]
Since $\tau_2\!>\!\tau_1$ and $\theta_2\!>\!\theta_1$ thanks to~(\ref{fourx11}c), $\text{sgn}\big(\mathcal{D}(\tau^{-1}_1)\big)\!=\!\text{sgn}(\theta_1\!-\!\tau_1)$ and $\text{sgn}\big(\mathcal{D}(\theta^{-1}_1)\big)\!=\!-\text{sgn}(\theta_1\!-\!\tau_1)$. As a result, we deduce that 
\[
\tau^{\sigma}_1 = z_1^{-1} \in [\tau_1,\theta_1], \qquad 
\tau^{\sigma}_2 = z_2^{-1} > \text{max}(\tau_1,\theta_1)
\]
due to the fact that $z\mapsto \mathcal{D}(z)$ is a convex second-degree polynomial, see Fig.~\ref{convex}. In most practical situations one expects to have $\theta_1\!>\!\tau_1$, since the bulk  viscosity is typically negligible to that in shear.

\begin{figure}[ht]
\centering{\includegraphics[width=0.63\linewidth]{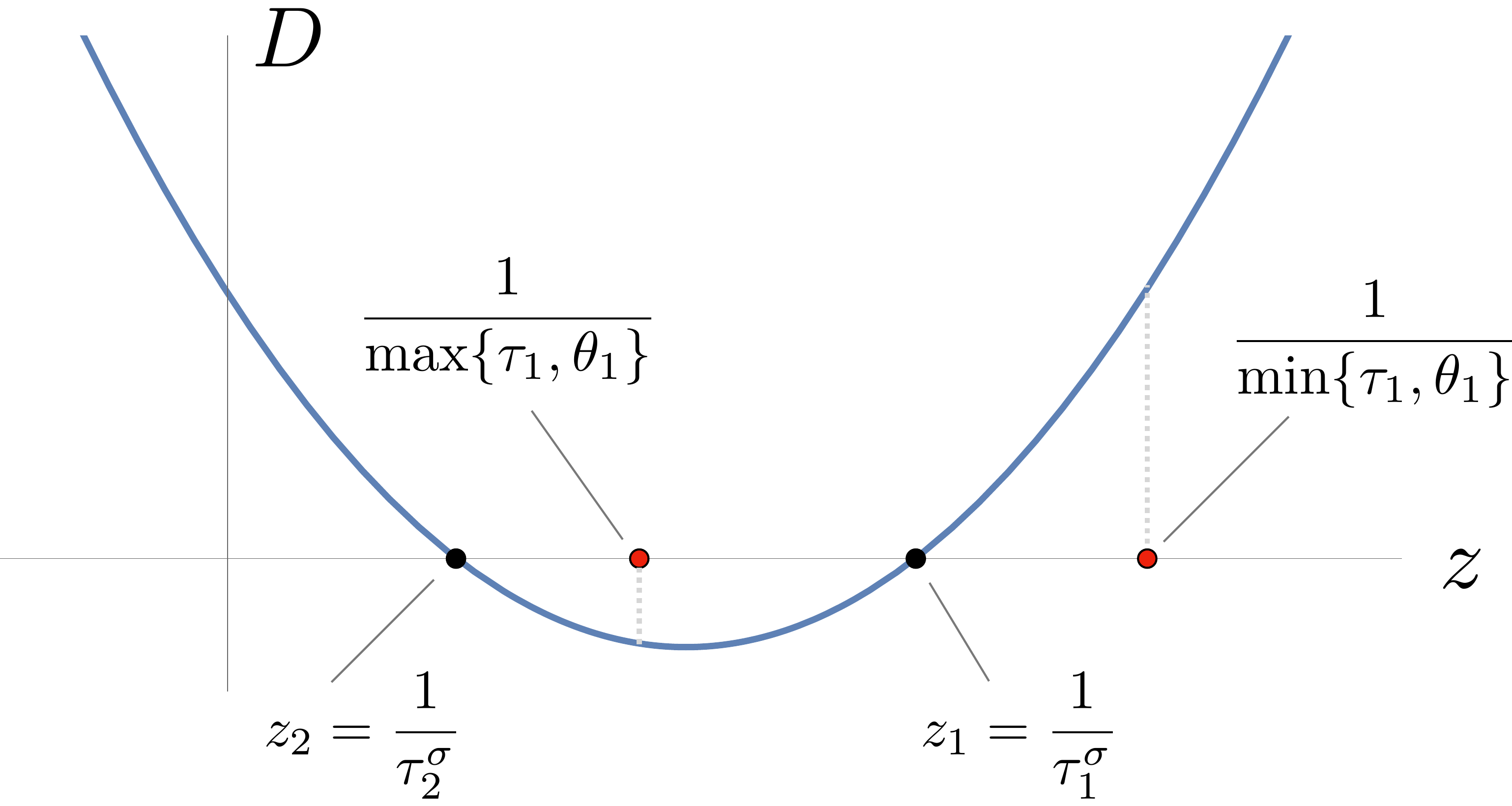}}
\caption{Relationship between $(\tau_1^\sigma,\tau_2^\sigma)$ and $(\tau_1,\theta_1)$ assuming $\tau_2\!>\!\tau_1$ and $\theta_2\!>\!\theta_1$.}
\label{convex}
\end{figure}

\subsection{Zener solid} 

For this visoelastic model, the two ``springs" ($G_1$ and $G_2$) and a single ``dashpot'' ($\eta$) are arranged as in Fig.~\ref{viscomod1}(d). On defining $\gamma,\tau_1$ and~$\tau_2$ as in~(\ref{fourx11}d), the Zener model obeys the same differential equation~\eqref{four12} as the SLS model, and all results obtained for the plane-stress behavior of the latter apply without modification to the Zener solid.

\subsection{Jeffreys material} \label{secJEF}

We next examine the Jeffreys viscoelastic model, whose three-parameter mechanical analogue is shown in Fig.~\ref{viscomod2}(a). Letting
\[
\tau_1 := \frac{\eta_1\!+\eta_2}{G}, \qquad \tau_2:= \frac{\eta_2}{G},
\]
for this model we obtain  
\begin{multline} \label{jeff1}
\sigma(t) + \tau_1\hh\dot\sigma(t) = \eta_1(\dot{\eps}(t) + \tau_2\hh\ddot{\eps}(t)), \quad \tau_1>\tau_2 ~~\Rightarrow~~ \\ 
\Rightarrow~~ \hat{C}(\omega)
= -\ii \eta_1 \omega \hh \frac{1-\ii\tau_2\hh\omega}{1-\ii\tau_1\hh\omega}
~~\Rightarrow~~
  C(t) = \eta_1\Big[\frac{\tau_2}{\tau_1} \delta'(t) + \frac{\tau_1-\tau_2}{\tau_1^2} \delta(t) - \frac{\tau_1-\tau_2}{\tau_1^3} \hh H(t) e^{-t/\tau_1} \Big].
\end{multline}
\begin{figure}[h!]
\centering{\includegraphics[width=0.67\linewidth]{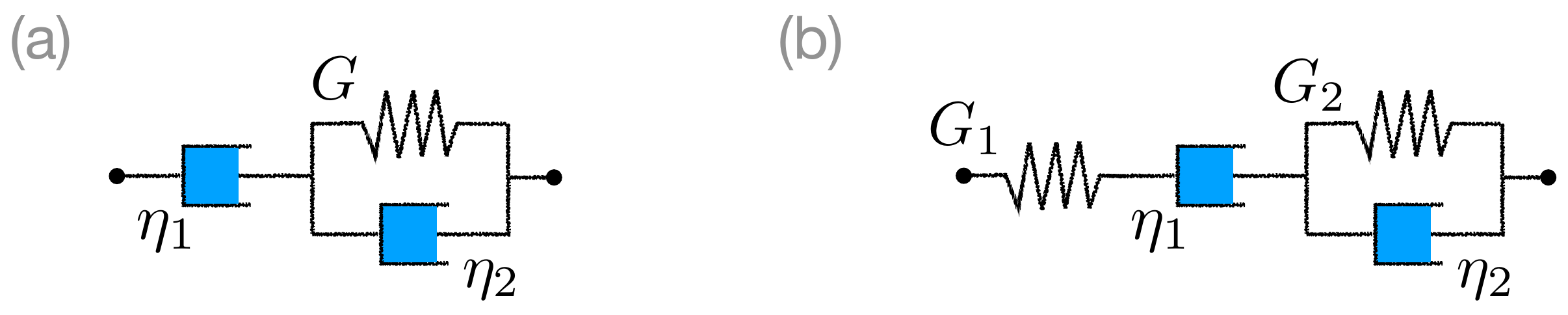}}
\caption{Extended mechanical models of the viscoelastic behavior: (a) Jeffreys material, and (b) Burgers material.} \label{viscomod2}
\end{figure}

As a result, one can write the complex bulk and shear modulus of an isotropic Jeffreys material as 
\[
\hkaC(\omega) = -\ii\eta_1\omega\hh\frac{1-\ii\tau_2\hh\omega}{1-\ii\tau_1\hh\omega}, \qquad
\hmuC(\omega) = -\ii\xi_1\omega\hh\frac{1-\ii\theta_2\hh\omega}{1-\ii\theta_1\hh\omega}, 
\]
where $\{\eta_1,\tau_1\!=\!(\eta_1\!+\eta_2)/\kappa,\tau_2\!=\!\eta_2/\kappa\}$ and~$\{\xi_1,\theta_1\!=\!(\xi_1\!+\xi_2)/G,\theta_2\!=\!\xi_2/G\}$ denote the ``outer'' damping parameter and relaxation time in bulk and shear, respectively. The plane-stress complex bulk modulus $\hkaCS(\omega)$ due to~\eqref{isovis8} is then obtained as
\[
\hkaCS(\omega)
= -\ii\eta^{\sigma}\!\omega\,\frac{(1-\ii\tau_2\omega)(1-\ii\theta_2\omega)}{A(\ii\omega)^2-B(\ii\omega)+1}
= \eta^{\sigma} \Big(\! -\frac{\tau_2\theta_2}{A}\ii\omega
+ \alpha_1 + \alpha_2
- \frac{\alpha_1}{1-\ii\tau^{\sigma}_1\omega} - \frac{\alpha_2}{1-\ii\tau^{\sigma}_2\omega} \Big), 
\]
where $A$ and~$B$ are constants given by~\eqref{aux05} with $(\kappa,\gamma)$ superseded by $(\eta_1,\xi_1)$; the plane-stress damping parameter $\eta^{\sigma}$ reads
\[
\eta^{\sigma} = \frac{9\eta_1\xi_1}{3\eta_1\!+4\xi_1};
\]
the constants $\alpha_{1/2}$ are given by
\begin{equation}  \label{alpha12J:defx}
\alpha_j
 = \frac{(\tau^{\sigma}_j-\tau_2)(\tau^{\sigma}_j-\theta_2)}{(\tau^{\sigma}_j)^2(\tau^{\sigma}_j-\tau^{\sigma}_{3-j})}, ~~~j=\overline{1,2};
\end{equation}
and the relaxation times $\tau^{\sigma}_1<\tau^{\sigma}_2$ are again given by $\tau^{\sigma}_i = z^{-1}_i$ with $z_1>z_2>0$ being the roots of~\eqref{D=0}. 
As in the SLS case, we observe that $z_{1/2}$ are necessarily (i) real-valued due to the fact that $\tau_1\!>\!\tau_2$ and $\theta_1\!>\!\theta_2$ which ensures that~\eqref{AB:discr} again holds (upon replacing $(\kappa,\gamma)$ with $(\eta_1,\xi_1)$),
and (ii) positive thanks to the positivity of~$A$ and~$B$. This yields the (causal) plane-stress memory function in bulk as
\begin{equation} \label{memory-jeff}
\kaCS(t) = \eta^{\sigma} \Big[\frac{\tau_2\theta_2}{A} \delta'(t) + (\alpha_1 + \alpha_2)\delta(t)
-\frac{\alpha_1}{\tau^{\sigma}_1} H(t) \, e^{-t/\tau^{\sigma}_1}
- \frac{\alpha_2}{\tau^{\sigma}_2} H(t) \, e^{-t/\tau^{\sigma}_2} \Big].
\end{equation}

\subsection{Burgers material} 

Finally, we consider the Burgers viscoelastic model shown in Fig.~\ref{viscomod2}(b) that is described by 
\begin{multline} \label{burg1}
G_2\hh \sigma(t) + [G_1\tau_1+G_2(\tau_1\nes+\tau_2)]\hh\dot\sigma(t) + G_2\tau_1\tau_2 \,\ddot\sigma(t) = G_1 G_2 \tau_1 \hh [\eps(t) + \tau_2\hh\dot\eps(t)] ~~ \Rightarrow \\ 
\hat{C}(\omega)
= \frac{G_1 G_2 \hh\tau_1\hh \omega \hh(1\!-\!\ii\tau_2\hh\omega)}{G_1\tau_1\hh \omega + \ii\hh G_2 (1\!-\!\ii\tau_1\hh\omega)(1-\ii\tau_2\hh\omega)} ~~\Rightarrow~~ \\
C(t) = G_1 \tau_1 \Big[\! (\alpha_1 + \alpha_2)\delta(t) 
-\frac{\alpha_1}{\tauT_1} H(t) \, e^{-t/\tauT_1}
- \frac{\alpha_2}{\tauT_2} H(t) \, e^{-t/\tauT_2} \Big], 
\end{multline}
where $\tau_j\!=\!\eta_j/G_j$ $(j\!=\!\overline{1,2})$ are the characteristic times; $\alpha_j$ are given by~\eqref{alpha12J:defx} with~$\theta_2\!=\!0$, and $\tauT_1<\tauT_2$ are given by $\tauT_i = 1/z_i$, with  $z_1>z_2>0$ being the real roots of~\eqref{D=0} with 
\begin{equation} \label{AB:BurgersK}
A = \tau_1\hh\tau_2>0, \qquad B = \frac{G_1\nes\tau_1 + G_2(\tau_1\nes+\tau_2)}{G_2}>0.
\end{equation}

In this vein, we consider an isotropic viscoelastic material with the complex bulk and shear modulus 
\begin{equation} \label{burg2}
\hkaC(\omega)= \frac{\kappa_1 \kappa_2 \hh\tau_1\hh \omega \hh(1\!-\!\ii\tau_2\hh\omega)}{\kappa_1\tau_1\hh \omega + \ii\hh \kappa_2 (1\!-\!\ii\tau_1\hh\omega)(1-\ii\tau_2\hh\omega)}, \quad 
\hmuC(\omega)= \frac{\gamma_1 \gamma_2 \hh\theta_1\hh \omega \hh(1\!-\!\ii\theta_2\hh\omega)}{\gamma_1\theta_1\hh \omega + \ii\hh \gamma_2 (1\!-\!\ii\theta_1\hh\omega)(1-\ii\theta_2\hh\omega)}, \quad 
\end{equation}
featuring the elastic moduli $\kappa_j$ (resp.~$\gamma_j)$ and characteristic times $\tau_j\!=\!\eta_j/\kappa_j$ (resp.~$\theta_j\!=\!\xi_j/\gamma_j$) in bulk (resp.~shear). By virtue of~\eqref{isovis8} and~\eqref{burg2}, the plane-stress complex bulk modulus is computed as 
\begin{equation} \label{burg3}
\hkaS(\omega) = \eta^\sigma \hh 
\frac{({3\eta_1+4\xi_1})\hh \kappa_2 \hh \gamma_2 \hh (\ii \omega) \hh \mathcal{K}_2 \hh \mathcal{M}_2}
{\kappa_2\hh\gamma_2 \hh\mathcal{K}_2 \hh \mathcal{M}_2 (4 \xi_1 \hh \mathcal{K}_1  + 3 \eta_1 \hh \mathcal{M}_1) - 
\eta_1 \xi_1 \hh(\ii \omega) \hh(3\kappa_2\hh \mathcal{K}_2 + 4\gamma_2\hh \mathcal{M}_2)}
  \,:=\, \eta^\sigma\,\frac{\mathcal{N}(\ii\omega)}{\mathcal{D}(\ii\omega)},
\end{equation}
where the (plane-stress) damping parameter $\kaS$ and affine functions $\mathcal{K}_j(\ii\oo)$ and~$\mathcal{M}_j(\ii\oo)$ are given by 
is given by
\begin{equation}\label{burg4}
\eta^\sigma := \frac{9\hh\eta_1\hh \xi_1}{3\eta_1+4\xi_1}, \qquad
\mathcal{K}_j(\ii\omega) = \ii\tau_j\hh\omega - 1, \quad \mathcal{M}_j(\ii\omega) = \ii\theta_j\hh\omega - 1 \quad (i=1,2).
\end{equation}
From~\eqref{isovis8} and~\eqref{burg2}, one may also observe that~$\mathcal{D}(z)$ is a cubic polynomial admitting the form
\[
\Dcal(z) = C_1 (1-\tau_2\hh z)(1-\thetaT_1\hh z)(1-\thetaT_2\hh z) + C_2(1-\theta_1\hh z)(1-\tauT_1\hh z)(1-\tauT_2\hh z) 
\]
for some constants~$C_1\!<\!0$ and~$C_2\!<\!0$, where $\tauT_1\!<\!\tauT_2$ and $\thetaT_1\!<\!\thetaT_2$ are the characteristic times inherent to~$\hkaC$ and~$\hmuC$, respectively. As a result, we obtain
\begin{equation}
\begin{aligned}
  \Dcal(\tauT_1^{-1})
 &= C(\tau_2\shm\tauT_1)(\tauT_1\shm\thetaT_1)(\tauT_1\shm\thetaT_2), \quad &
  \Dcal(\thetaT_1^{-1})
 &= C(\theta_2\shm\thetaT_1)(\thetaT_1\shm\tauT_1)(\thetaT_1\shm\tauT_2), \\
  \Dcal(\tauT_2^{-1}),
 &= C(\tau_2\shm\tauT_2)(\tauT_2\shm\thetaT_1)(\tauT_2\shm\thetaT_2) \quad &
  \Dcal(\thetaT_2^{-1})
 &= C(\theta_2\shm\thetaT_2)(\thetaT_2\shm\tauT_1)(\thetaT_2\shm\tauT_2),
\end{aligned} \label{aux04}
\end{equation}
where $C$ is for each evaluation a positive constant. From~\eqref{D=0} and expressions for~$A$ and~$B$ due to~\eqref{AB:BurgersK} written for the complex bulk and shear modulus, on the other hand, an examination similar to that illustrated in Fig.~\ref{convex} demonstrates that $\tauT_2>\max\{\tau_1,\tau_2\}>\min\{\tau_1,\tau_2\}>\tauT_1$ and similarly $\thetaT_2>\max\{\theta_1,\theta_2\}>\min\{\theta_1,\theta_2\}>\thetaT_1$. As a result, we have
\begin{equation}
(\tau_2\shm\tauT_1) > 0, \quad (\tau_2\shm\tauT_2) < 0, \quad
(\theta_2\shm\thetaT_1) > 0, \quad (\theta_2\shm\thetaT_2) < 0. \label{aux07}
\end{equation}
Letting $t_1,t_2,t_3,t_4$ denote the characteristic times $\tauT_1,\tauT_2,\thetaT_1,\thetaT_2$ ordered along increasing values, a careful inspection of~\eqref{aux04} for each possible ordering of $\tauT_1,\tauT_2,\thetaT_1,\thetaT_2$ constrained by $\tauT_1\!<\!\tauT_2$, $\thetaT_1\!<\!\thetaT_2$,  and~\eqref{aux07} demonstrates that
\begin{equation}
  D(t^{-1}_4) < 0, \quad~ D(t^{-1}_3) > 0, \quad~ D(t^{-1}_2) < 0, \quad~ D(t^{-1}_1) > 0.
\end{equation}
The polynomial $D(z)$ thus has a root $z_i\!\in\!\mathbb{R}$ ($i\!=\!\overline{1,3}$) in each interval $z\in(t_{i+1}^{-1},t_i^{-1})$. Since $D(z)$ is cubic, all of its roots are accounted for, and each $z_i$ (real, positive, and of single multiplicity) defines a plane-stress relaxation time $\tau^{\sigma}_i=1/z_i\in[t_i,t_{i+1}]$. In other words, the four relaxation times of the Burgers model in bulk and shear ($\tilde\tau_1,\tilde\tau_2,\tilde\theta_1$ and~$\tilde{\theta}_2$) generate the three relaxation times of the plane-stress bulk memory function. Using those roots and deploying the partial fraction expansion of $\Ncal(\ii\oo)/\Dcal(\ii\oo)$, the plane-stress bulk modulus $\hkaS(\omega)$ is finally written as
\[
\hkaCS(\omega) = \eta^\sigma \hh \Big(\alpha_1+\alpha_2+\alpha_3 - \frac{\alpha_1}{1-\ii\tau^{\sigma}_1\omega} - \frac{\alpha_2}{1-\ii\tau^{\sigma}_2\omega} - \frac{\alpha_3}{1-\ii\tau^{\sigma}_3\omega} \Big), \qquad
\alpha_i = 
\frac{(\tau^{\sigma}_i)^2 - (\tau_2\!+\theta_2)\tau^{\sigma}_i + \tau_2\hh\theta_2}
{\tau^{\sigma}_i\hh(\tau^{\sigma}_i-\tau^{\sigma}_j)(\tau^{\sigma}_k-\tau^{\sigma}_i)}
\]
where $(i,j,k)$ are circular permutations of $(1,2,3)$. As in previous cases, this $\hkaCS(\omega)$ defines a causal memory function, which is easily found via the identities from Table~\ref{tab1} to read
\begin{equation}
\kaCS(t) =\eta^\sigma \hh\Big[ 
(\alpha_1+\alpha_2+\alpha_3)\delta(t) - \frac{\alpha_1}{\tau^{\sigma}_1} H(t)\,e^{-t/\tau^{\sigma}_1}
- \frac{\alpha_2}{\tau^{\sigma}_2} H(t)\,e^{-t/\tau^{\sigma}_2} - \frac{\alpha_3}{\tau^{\sigma}_3} H(t)\,e^{-t/\tau^{\sigma}_3} \Big], 
\end{equation}
noting in particular that $\alpha_1+\alpha_2+\alpha_3\!=\!-\tau_2\hh\theta_2/(\tau^{\sigma}_1\hh \tau^{\sigma}_2\hh\tau^{\sigma}_3)$. 

\section{Illustrations}
\label{sec:numex}

\noindent From the foregoing analysis, we observe that an isotropic viscoelastic solid described by any of the elementary models shown in Fig.~\ref{viscomod1} features an \emph{effective elastic} bulk modulus~$\kappa$ and shear modulus~$\gamma$. In this vein we conveniently introduce the effective Poisson's ratio
\[
\nu := \frac{3\kappa-2\gamma}{2(3\kappa+\gamma)},
\]
attributable the Kelvin-Voigt, Maxwell, SLS, and Zener isotropic solids. For any of the forgoing models, the \emph{plane-stress elastic} bulk modulus~$\kappa^\sigma$ is given by~\eqref{ksig:elas} and can be rewritten as 
\[
\frac{\,\kaS}{\!\kappa} = \frac{9}{3\varrho+4}, \qquad \varrho := \frac{\kappa}{\gamma} = \frac{2(1\nes+\nes\nu)}{3(1\nes-\nes2\nu)}. 
\]
In support of the ensuing examples, Fig.~\ref{num1} plots the ratio~$\kaS/\kappa$ versus~$\rho$ and~$\nu$ under the premise of non-auxetic (effective) elastic behavior. As can be seen from the display, $0\!<\!\kappa^\sigma/\kappa \nes\leqslant\!1.5$ is a monotonically-decreasing function of either argument, noting in particular that $\kappa^\sigma\!\nes=\nes\kappa$ for $\kappa\nes=\nes\tfrac{5}{3}\hh\gamma$, i.e.~$\nu\nes=\nes0.25$. 
\begin{figure}[ht]
\centering{\includegraphics[width=0.85\linewidth]{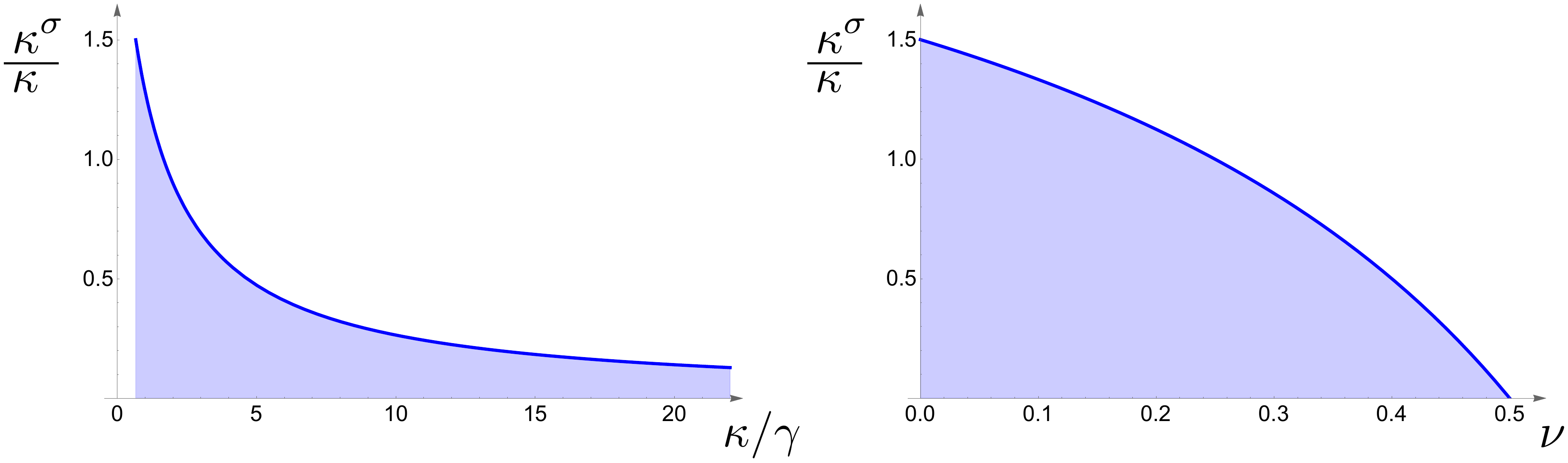}}
\caption{Variation of the plane-stress elastic bulk modulus, $\kappa^\sigma$, featured by the elementary viscoelastic models in  Fig.~\ref{viscomod1}.}
\label{num1}
\end{figure}

\begin{table}[t]
\centering \renewcommand{\arraystretch}{1.1}
\begin{tabular}{|c||c|c|c|c||c|c|c|c|}\hline
Model & \multicolumn{4}{|c||}{Bulk parameters} & \multicolumn{4}{c|}{Shear parameters} \\ \hline\hline
 & \;$K_1$\; & ~$\eta_1$~ & \;$K_2$\; & ~$\eta_2$~ & \;$G_1$\; & ~$\xi_1$~ & \;$G_2$\; & ~$\xi_2$~ \\ \hline 
Kelvin-Voigt & 2 & 0.1 & \cellcolor{gray!15} & \cellcolor{gray!15} & 1 & 0.5 & \cellcolor{gray!15} & \cellcolor{gray!15} \\ \hline 
Maxwell      & 2 & 0.1 & \cellcolor{gray!15} & \cellcolor{gray!15} & 1 &  1  & \cellcolor{gray!15} & \cellcolor{gray!15} \\ \hline 
SLS          & 2 & 0.02 &1  & \cellcolor{gray!15} & 1 & 0.5 & 0.1 & \cellcolor{gray!15} \\ \hline
Jeffreys    & 2 & \cellcolor{gray!15} & \cellcolor{gray!15}& \cellcolor{gray!15} & 0.5 & 1 & \cellcolor{gray!15} & 0.1 \\ \hline
Burgers      & 4 & \cellcolor{gray!15}& \cellcolor{gray!15}& \cellcolor{gray!15} & 3 & 1 & 0.5 & 0.1 \\ \hline
\end{tabular} \renewcommand{\arraystretch}{1}
\caption{Example parameters of an isotropic viscoelastic solid described by the constitutive models shown in Fig.~\ref{viscomod1} and Fig.~\ref{viscomod2}. In the absence of parameter~$\chi_2$ ($\chi\!\in\{K,\eta,G,\xi\}$), $\chi_1\mapsto\chi$. }
\label{tabn1}
\end{table}

\subsection{Kelvin-Voigt solid}

As examined in Section~\ref{kelvin}, the plane-stress complex bulk modulus $\hkaCS$ of an isotropic Kelvin-Voigt solid is computed via~\eqref{ksig:comp} in terms of~$\kaS$, the bulk and shear relaxation times $\tau$ and $\theta$, and the plane-stress relaxation time~$\tau^\sigma$. From~\eqref{ksig:elas}, we obtain 
\begin{equation}\label{char-kelv}
\frac{\tau^{\sigma}}{\tau} = \frac{3\varrho\hh + 4 \zeta}{3\varrho+4}, \qquad \zeta := \frac{\theta}{\tau} 
\end{equation}
which uniquely specifies the ratio~$\tau^{\sigma}/\tau$ in terms of~$\varrho$ and~$\zeta$ as shown in Fig.~\ref{num2-kelv}(a). From the display, one observes that for given bulk relaxation time~$\tau$, its plane-stress companion~$\tau^{\sigma}$ increases with (i) increasing shear relaxation time~$\theta$ and (ii) decreasing $\kappa/\gamma$, i.e. diminishing (effective) Poisson's ratio. In realistic situations where $\zeta>1$, it is also apparent from both~\eqref{char-kelv} and Fig.~\ref{num2-kelv}(a) that $\tau^{\sigma}>\tau$, a feature that is reflected in the reminder of this section. The behavior of the complex moduli $\hmuCS\nes=\nes\muC$, $\hkaC$ and~$\hkaCS$ is  illustrated in Fig.~\ref{num2-kelv}(b,c) assuming the model parameters listed in Table~\ref{tabn1}, which yield~$\kappa\!=\!2$, $\gamma\!=\!1$, $\tau\!=\!0.05$, and~$\theta\!=\!0.5$ in~\eqref{ksig:comp}--\eqref{ksig:elas}.  Beyond the increase in both real and imaginary parts of the complex bulk modulus at higher frequencies ($\omega\!>\!2$), we observe that $\hkaCS(\omega)$ does not adhere to the Kelvin-Voigt model and, in fact, to any of the phenomenological descriptions examined in this work. 
\begin{figure}[ht]
\centering{\includegraphics[width=0.92\linewidth]{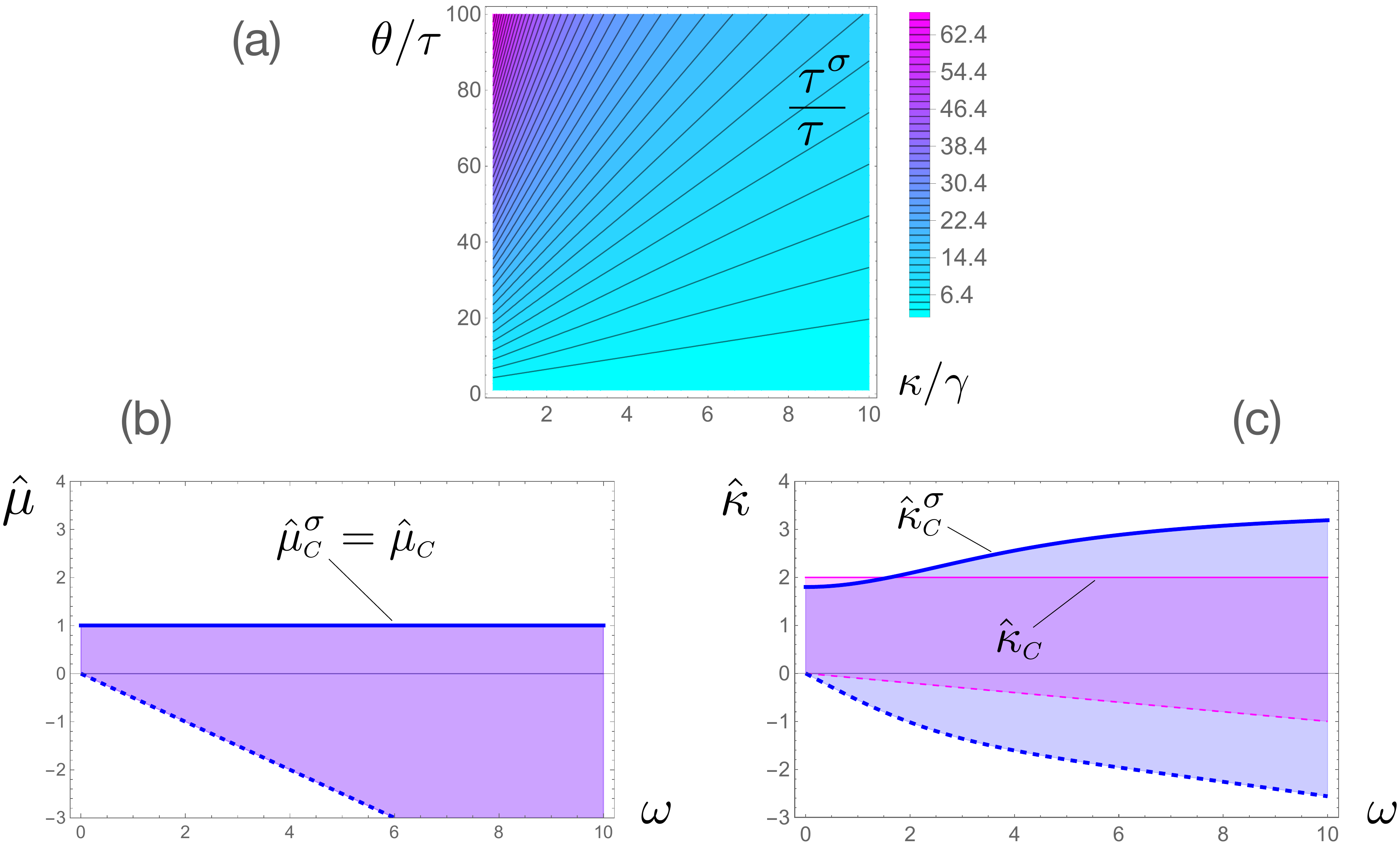}}
\caption{Kelvin-Voigt solid: plane-stress (a) relaxation time~$\tau^\sigma$, (b) complex shear modulus~$\hmuCS$, and (c) complex bulk modulus~$\hkaCS$.}
\label{num2-kelv}
\end{figure}

\subsection{Maxwell material}

From~\eqref{taus-maxw}, the plane-stress relaxation time~$\tau^\sigma$ featured by an isotropic Maxwell material can be rewritten as 
\[
\frac{\tau^{\sigma}}{\tau} = \frac{3\varrho\hh + 4 \zeta}{3\varrho+4}, \qquad \zeta := \frac{\theta}{\tau}; 
\]
a dependence that, for $\zeta\!>\!1$ and given~$\tau$, features monotonic growth of $\tau^\sigma$ with both~$\rho$ and~$\zeta$ as shown in Fig.~\ref{num3-maxw}(a). The frequency dependence of the complex moduli~$\hmuCS$ and~$\hkaCS$ is  illustrated in Fig.~\ref{num3-maxw}(b,c) assuming the model parameters listed in Table~\ref{tabn1}, which correspond to~$\kappa\!=\!2$, $\gamma\!=\!1$, $\tau\!=\!0.05$, and~$\theta\!=\!1$ in~\eqref{kapas-maxw}--\eqref{taus-maxw}. The result illustrates the fact that~$\hkaCS(\omega)$ is also described by the Maxwell model, endowed with relaxation time $\tau<\tau^\sigma<\theta$. In fact, this is the only case encountered in this study where the complex bulk modulus~$\hkaC$ and its plane-stress counterpart~$\hkaCS$ share the phenomenological description.
\begin{figure}[ht]
\centering{\includegraphics[width=0.92\linewidth]{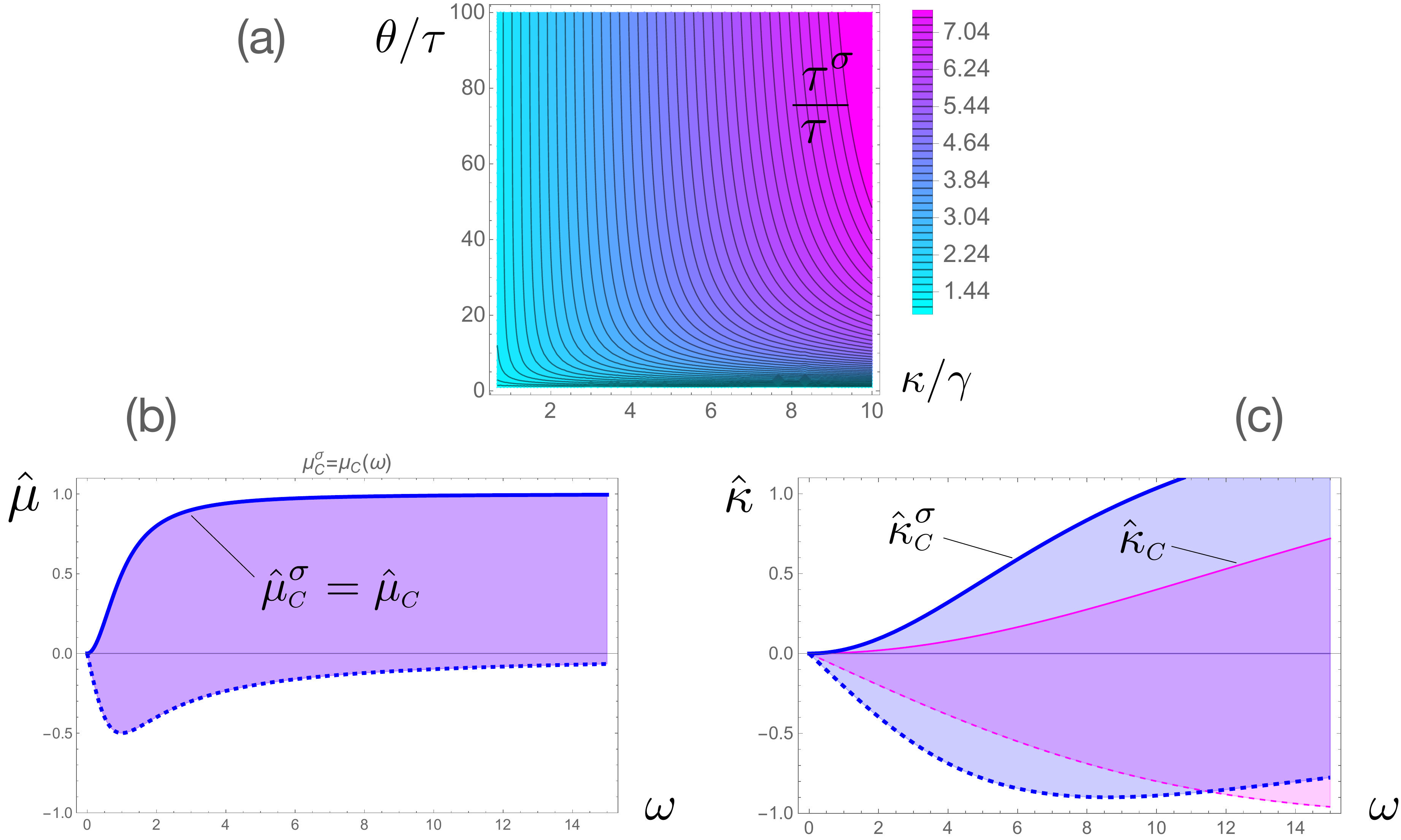}}
\caption{Maxwell material: plane-stress (a) relaxation times~$\tau^\sigma$, (b) complex shear modulus~$\hmuCS$, and (c) complex bulk modulus~$\hkaCS$  (solid lines: real parts, dashed lines: imaginary parts).}
\label{num3-maxw}
\end{figure}

\subsection{Standard linear solid}

From the analysis in Section~\ref{secSLS}, we find that the reciprocals of the plane-stress (bulk) relaxation times, $(\tau_j^\sigma)^{-1}$, solve the quadratic equation~\eqref{D=0} whose coefficients~$A$ and~$B$ are given by~\eqref{aux05}. To facilitate the illustration, we rewrite these coefficients as 
\begin{equation} \label{aux05x}
A = \tau_{\nes\theta}^2 \, \frac{3\varrho+4\tauT\hh\thetaT}{3\varrho+4}, \qquad
B = \tau_{\nes\theta} \, \frac{3\varrho\hh(\tilde{\upsilon}+\tilde{\upsilon}^{-1})+4(\tauT\hh \tilde{\upsilon}\!+\thetaT\hh \tilde{\upsilon}^{-1})}{3\varrho+4},
\end{equation}
where
\begin{equation} \label{aux05y}
\tau_{\nes\theta} := \sqrt{\tau_2\hh\hh\theta_1}, \qquad \tauT := \frac{\tau_1}{\tau_2}, \qquad \thetaT := \frac{\theta_2}{\theta_1}, \qquad \tilde{\upsilon} := \sqrt{\frac{\tau_2}{\theta_1}}. 
\end{equation}
For given~$\varrho$, this result exposes~$\tau_j^\sigma/\tau_{\theta}$ ($j\!=\!\overline{1,2}$) as a function of~$(\tauT,\thetaT,\tilde{\upsilon})$ under the premise $\tau_\theta\!>\!0$. Letting $\varrho\!=\!7.\bar{3}$ ($\nu\!\simeq\!0.435$), we illustrate this variation in Fig~\ref{num4-slss}(a) over the box 
\[
\Omega=\big\{(\tauT,\thetaT,\tilde{\upsilon}): 0<\!\tauT\!<1,~ 1<\!\thetaT\!<10,~ 0<\!\tilde{\upsilon}\!<1\big\}, 2x
\]
which guarantees that $\max\{\tau_1,\tau_2\}<\min\{\theta_1,\theta_2\}$, i.e.~that the bulk viscosity of an SLS is inferior to that in shear. The variation of the complex moduli~$\hmuCS$ and~$\hkaCS$ is exemplified in Fig.~\ref{num4-slss}(b,c) assuming the model parameters listed in Table~\ref{tabn1}, which yield~$\kappa\!=\!0.\overline{6}$, $\gamma\!=\!0.\overline{09}$ ($\varrho\!=\!7.\bar{3}$), $\tau_1\!=\!0.00\overline{6}$, $\tau_2\!=\!0.02$, $\theta_1\!=\!0.\overline{45}$ and~$\theta_2\!=\!5$ in~\eqref{cmodsls}--\eqref{D=0}. Qualitatively speaking, the behavior of~$\hkaCS(\omega)$ due to the SLS model features a mixture of trends furnished by the previous two models. Specifically, we observe that $\Re[\hkaCS]\lesseqgtrslant\Re[\hkaC]$ depending on the frequency (Kelvin-Voigt-type behavior), while $\Im[\hkaCS]$ is non-monotonic (Maxwell-type behavior).  
\begin{figure}[ht]
\centering{\includegraphics[width=0.92\linewidth]{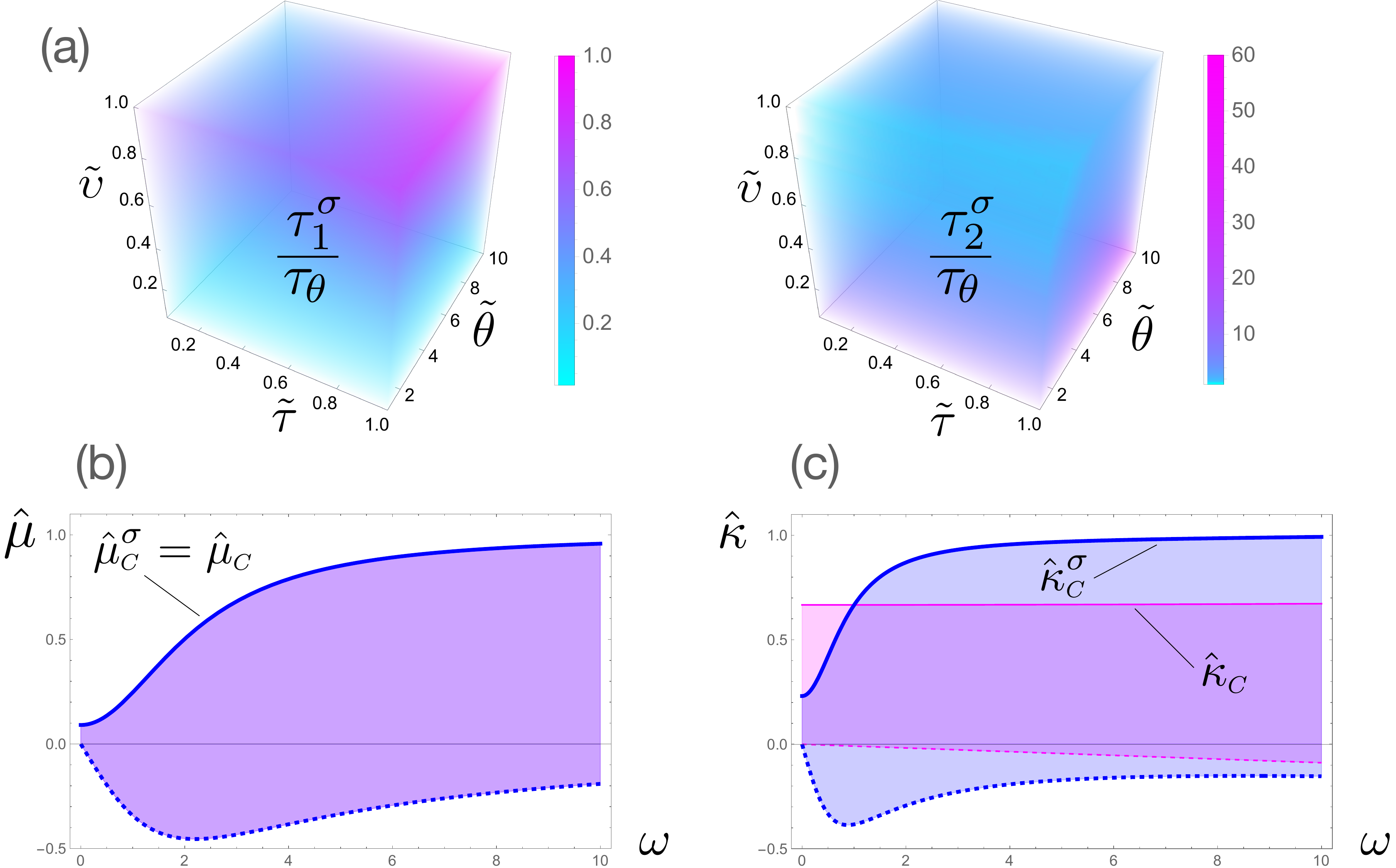}}
\caption{Standard linear solid: plane-stress (a) relaxation times~$\tau_j^\sigma$, (b) complex shear modulus~$\hmuCS$, and (c) complex bulk modulus~$\hkaCS$ (solid lines: real parts, dashed lines: imaginary parts).}
\label{num4-slss}
\end{figure}

For completeness, we note that the relaxation times~$\tau_j^\sigma$ synthesized in Fig.~\ref{num4-slss}(a) meet the constraints discussed in Section~\ref{secSLS}, see Fig.~\ref{convex}. Borrowing the values of $\kappa,\gamma,\tau_1$ and~$\theta_1$ from the last example, this is illustrated in  Fig.~\ref{convex2} which demonstrates that $\tau_1<\tau_1^\sigma<\theta_1<\tau_2^\sigma$ when $\tau_2$ and~$\theta_2$ are varied. 
\begin{figure}[ht]
\centering{\includegraphics[width=0.6\linewidth]{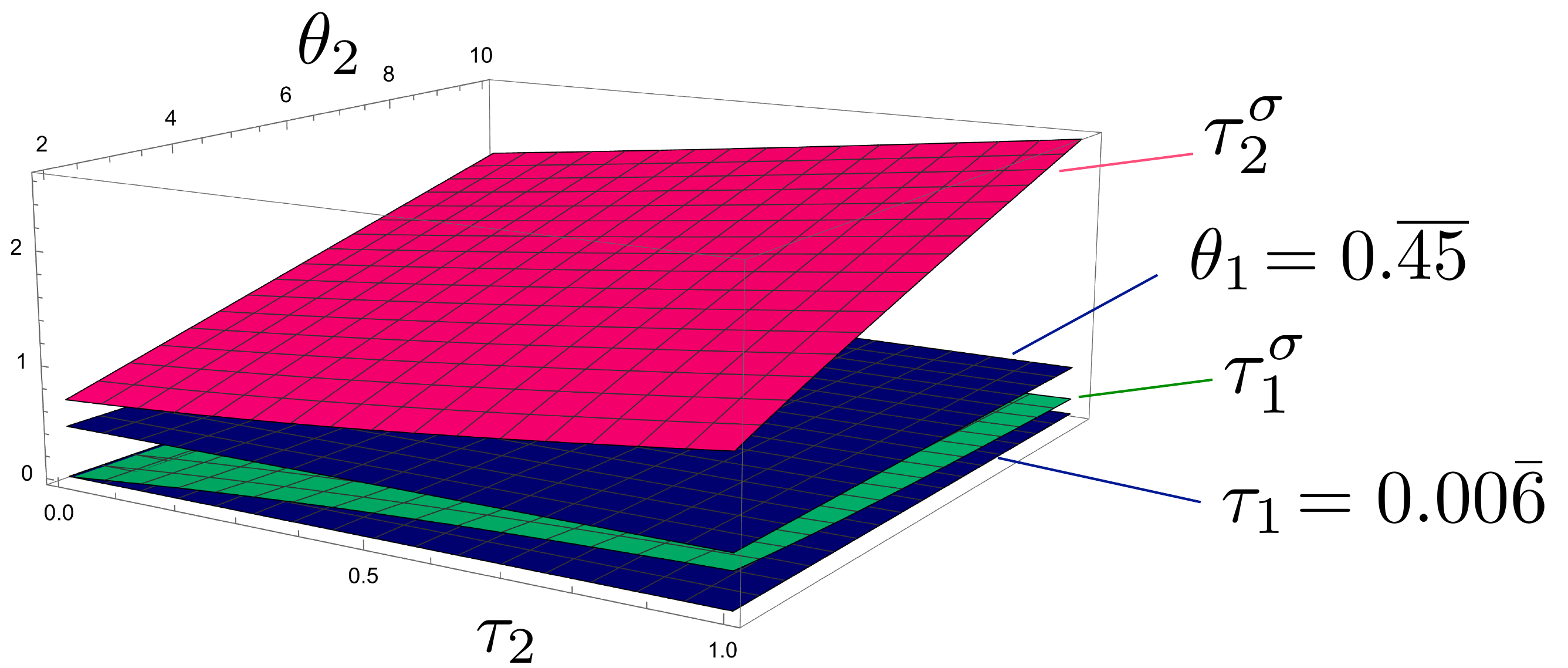}}
\caption{Standard linear solid: example behavior of $\tau_j^\sigma$ ($j\!=\!\overline{1,2}$) for fixed $\tau_1$ and~$\theta_1$.}
\label{convex2}
\end{figure}

\subsection{Jeffreys material}

With reference to Table~\ref{tabn1}, in this example we assume \emph{purely elastic} behavior in bulk and Jeffreys behavior in shear, namely
\begin{multline} \label{jeffn1}
\hkaC(\omega) = \kappa, \quad~
\hmuC(\omega) = -\ii\xi_1\omega\hh\frac{1-\ii\theta_2\hh\omega}{1-\ii\theta_1\hh\omega}  
\quad \implies \quad \\
\hkaCS(\omega) = \eta^\sigma\hh\frac{3(-\ii \omega)(1-\ii\theta_2\omega)}{A(\ii\omega)^2-B(\ii\omega)+1} \,=\,
\eta^{\sigma} \Big(\alpha_1 + \alpha_2
- \frac{\alpha_1}{1-\ii\tau^{\sigma}_1\omega} - \frac{\alpha_2}{1-\ii\tau^{\sigma}_2\omega} \Big), 
\end{multline}
where~$\kappa\nes=\nes K$; $\theta_1 = (\xi_1\!+\xi_2)/\gamma$; $\theta_2 = \xi_2/\gamma$; $\gamma\nes=\nes G$; $\eta^\sigma\nes=\nes \xi_1$; constants $\alpha_j$ are given by~\eqref{alpha12J:defx} with~$\tau_2\nes=\nes0$, and  
\begin{equation} \label{jeffn2}
A = \tau_{\xi}^2 \, \frac{4\varrho \hh \zeta}{3}, \quad~ B = \tau_{\xi}\, \frac{3\varrho(\zeta\shp 1)\shp 4}{3} \qquad \text{with}~~~ \tau_{\xi}:= \frac{\xi_1}{\kappa}, ~~~ \zeta := \frac{\xi_2}{\xi_1}.
\end{equation}
Since the plane-stress relaxation times $\tau^\sigma_j$ are given by the reciprocals of the roots of~\eqref{D=0}, 
we observe that $\tau^\sigma_j$ are necessarily real-valued because
\[
B^2-4\hh A \,=\, \tfrac{\tau_{\xi}^2}{9}\big[(3\varrho(\zeta\shp 1)\shp 4)^2 - 12\varrho\hh\zeta\big] \,>\, 0 \quad \text{for}~\varrho>0 ~\text{and}~ \zeta>0.
\]
On the basis of~\eqref{jeffn1}--\eqref{jeffn2}, the variation of the scaled relaxation times~$\tau^{\sigma}_j/\tau_{\xi}$ is plotted in Fig.~\ref{num5-jeff}(a). As seen from the display, there is a separation in scale between the two parameters in that (for given~$\varrho$ and~$\zeta$) $\tau_1^\sigma \ll \tau_2^\sigma$. Assuming the model parameters listed in Table~\ref{tabn1}, the complex moduli~$\hmuCS$ and~$\hkaCS$ due to~\eqref{jeffn1} are illustrated in Figs.~\ref{num4-slss}(b,c) which demonstrate that for an elastic bulk modulus $\hkaC(\omega) = \kappa$, its plane-stress companion~$\hkaCS(\omega)$ largely features Burgers-like behavior, inherited from the material's viscosity in shear. 

\begin{remark}
On recalling Fig.~\ref{spring-pots} and Fig.~\ref{viscomod2}, it is useful to note that the complex bulk modulus $\hkaCS(\omega)$ due to~\eqref{jeffn1} is in fact described the 1D Burgers model connecting (in series) an elastic spring with modulus $\tfrac{9}{4}\kappa$ and the Jeffreys element with complex modulus~$3\hmuC(\omega)$. In this vein, the Jeffreys model results shown in Fig.~\ref{num5-jeff} simultaneously describe the plane-stress bulk behavior of an isotropic Burgers solid whose parameters are listed in Table~\ref{tabn1}. 
\end{remark}

\begin{figure}[ht]
\centering{\includegraphics[width=0.92\linewidth]{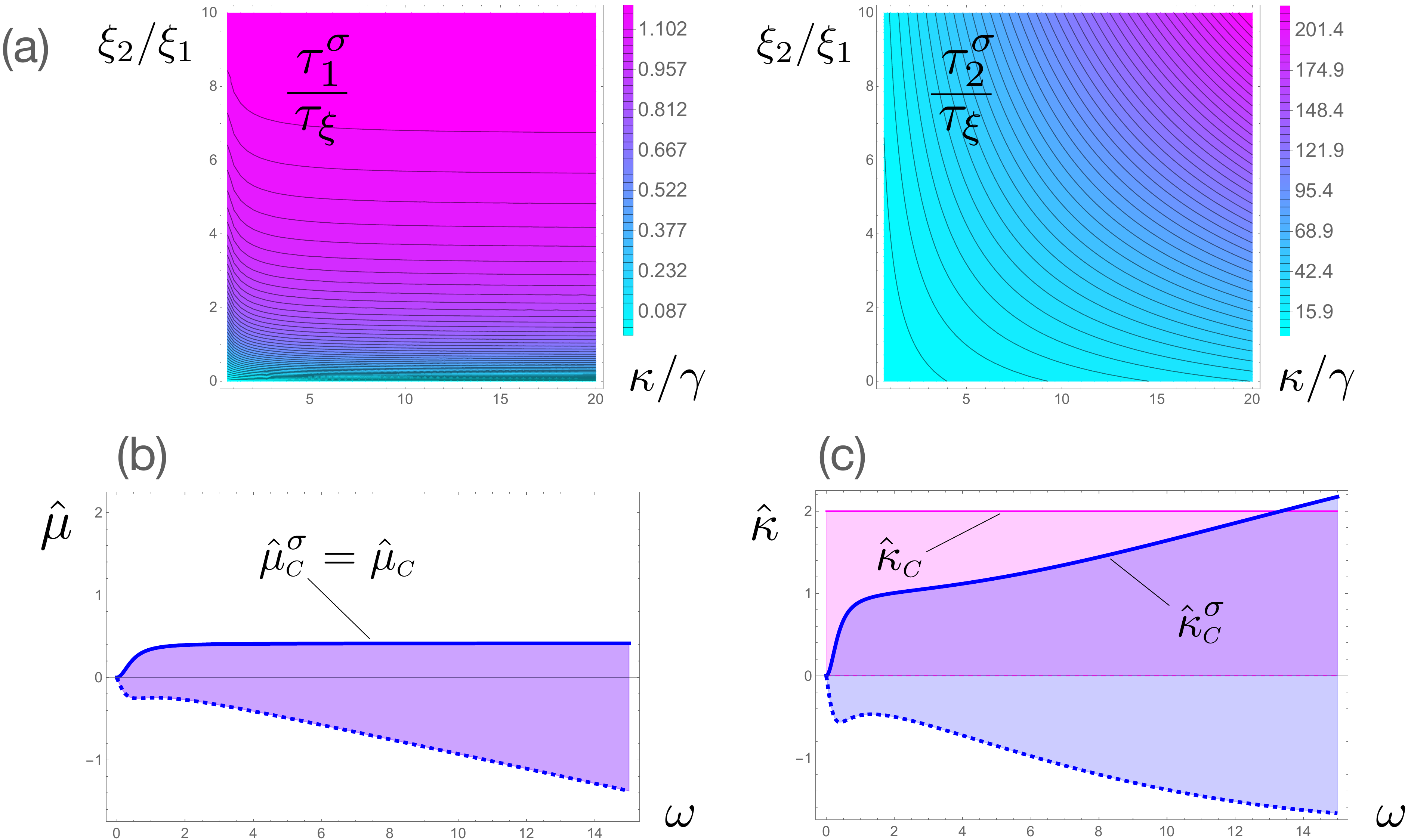}
\caption{Jeffreys material: plane-stress (a) relaxation times~$\tau_j^\sigma$, (b) complex shear modulus~$\hmuCS$, and (c) complex bulk modulus~$\hkaCS$ (solid lines: real parts, dashed lines: imaginary parts).}
\label{num5-jeff}}
\end{figure}

\subsection{Plane-stress thermodynamic potentials for a transversely isotropic viscoelastic solid}

The foregoing examples are uniformly focused on the case of material isotropy. For generality, however, it is also of interest to shed light on the anisotropic case. On recalling the analysis in Section~\ref{GSM}, we fill the gap by examining the case of a transversely anisotropic viscoelastic solid through the prism of plane-stress thermodynamic potentials featured by~\eqref{phi:psi:def-pstress1}. For clarity of discussion, we pursue the germane tensorial analysis using Voigt (matrix) notation summarized in Appendix~\ref{A:Voigt}. Specifically, for the stress tensor~$\bfsig$, strain tensor~$\bfeps$, and projection operators~$\bfP$ and~$\bfR$ we adopt the representation
\begin{equation} \label{voigtnot}
\begin{aligned}
\bfsigup &=\! \begin{array}{cccccc}
[\, \sigma_{11} & \sigma_{22} & \sigma_{33} & \sigma_{12} & \sigma_{23} & \sigma_{31} \,]\Tsup \end{array},  \qquad & \bfepsup &=\! \begin{array}{cccccc}
[\, \eps_{11} & \eps_{22} & \eps_{33} & 2\hh\eps_{12} & 2\hh\eps_{23} & 2\hh\eps_{31} \,]\Tsup \end{array}, \\
\bfsigup\psup  &=\! \begin{array}{cccccc}
[\, \sigma_{11} & \sigma_{22} & 0 & \sigma_{12} & 0 & 0 \,]\Tsup \end{array},  \qquad &\bfepsup\psup &=\! \begin{array}{cccccc}
[\, \eps_{11} & \eps_{22} & 0 & 2\hh\eps_{12} & 0 & 0 \,]\Tsup \end{array}, \\
\bbP &= \;\text{diag}\! \begin{array}{cccccc} [\, 1~~~  1~~~  0~~~  1~~~  0~~~ 0 \,] \end{array}, & 
\bbR &= \;\text{diag}\! \begin{array}{cccccc} [\, 0~~~  0~~~  1~~~  0~~~  1~~~ 1 \,] \end{array}.
\end{aligned}
\end{equation}

In what follows, we assume the constitutive behavior of a viscoelastic solid to exhibit axial symmetry with respect to the $\xi_3$-axis, and consider its plane stress reduction in the $\xi_3-\xi_1$ plane. For this type of material symmetry~\cite{ting:96}, the respective matrix counterparts $\bbC_\eps, \bbC_m$ and~$\bbC_\alpha$ of the fourth-order free energy tensors $\Ce, \CmT$ and~$\Ca$ in~\eqref{phi:psi:def} take the form
\begin{equation} \label{trans1}
\bbC_{\bullet} = \begin{bmatrix}
a_{\bullet} & b_{\bullet} & c_{\bullet} & 0           & 0           & 0  \\
b_{\bullet} & a_{\bullet} & c_{\bullet} & 0           & 0           & 0  \\
c_{\bullet} & c_{\bullet} & d_{\bullet} & 0           & 0           & 0  \\
0           & 0           & 0           & e_{\bullet} & 0           & 0  \\
0           & 0           & 0           & 0           & e_{\bullet} & 0  \\
0           & 0           & 0           & 0           & 0 & \tfrac{1}{2} (a_{\bullet}\!-b_{\bullet})
\end{bmatrix}, \quad \bullet\in\{\eps,m,\alpha\}. 
\end{equation}
By virtue of~\eqref{schur:def}, \eqref{td-pots-ps1} and~\eqref{schur:def2}, we obtain
\begin{equation} \label{trans2}
\bbC_{\bullet}\pssup = \left\lceil \begin{bmatrix}
a_\bullet + x_{\bullet} & b_\bullet + x_{\bullet}  & 0 \\
b_\bullet + x_{\bullet}  & a_\bullet + x_{\bullet}  & 0 \\
0 & 0 & e_{\bullet}
\end{bmatrix} \right\rceil, \quad \bullet\in\{\eps,\alpha\}, \qquad 
x_\bullet = \frac{c_m^2\hh d_\bullet + c_\bullet^2\hh d_\bullet' - 2 c_\bullet \hh c_m \hh d_m}{d_m^2-d_\alpha\hh d_\epsilon}, \quad
\left\{\begin{aligned}
 d_\epsilon' = d_\alpha \\
 d_\alpha' = d_\epsilon
\end{aligned} \right.
\end{equation}
and 
\begin{equation} \label{trans3}
\bbC_m\pssup = \left\lceil \begin{bmatrix}
a_m + y_m & b_m + y_m  & 0 \\
b_m + y_m  & a_m + y_m  & 0 \\
0 & 0 & e_m
\end{bmatrix} \right\rceil, \qquad 
y_m = \frac{c_m(c_\alpha\hh d_\eps + d_\alpha\hh c_\eps) - d_m(c_m^2\hh +c_\alpha\hh c_\eps)}
{d_m^2-d_\alpha\hh d_\epsilon}, 
\end{equation}
where for a generic $3\times 3$ matrix $\textbf{M}$, $\lceil \textbf{M} \rceil$ denotes its $6\times 6$ ``inflation'' generated by the insertion of zero rows and columns at locations (3,5,6) corresponding to the zeros of the 6-vectors $\bfsigup\psup$ and $\bfepsup\psup$ in~\eqref{voigtnot} (see Appendix~\ref{A:Voigt} for details). In the isotropic case, the result simplifies in that $d_\bullet\!=\!a_\bullet$, $c_\bullet\nes=\nes b_\bullet$, and $e_\bullet \nes=\! \tfrac{1}{2}(a_\bullet\!-\! b_\bullet)$ for $\bullet\!\in\!\{\eps,m,\alpha\}$ (two independent moduli per matrix). The corresponding matrices describing the plane-stress dissipation potential in~\eqref{phi:psi:def-pstress1} are obtained by replacing the symbol ``$\textbf{C}$'' by~``$\textbf{D}$'' in~\eqref{trans1}, \eqref{trans2}, and~\eqref{trans3}. 

\section{Summary}  

\noindent In this study, we expose the relationship between the ``bulk'' constitutive parameters of a viscoelastic solid and their plane-stress counterparts; a classical topic that has so far eluded an in-depth scrutiny. In doing so, we focus on the frequency-domain behavior of linear viscoelastic materials whose 3D rheological behavior is described by a set of constant-coefficient ordinary differential equations. For the special case of material isotropy, we provide an in-depth analysis of the complex moduli and memory functions of solids whose 3D bulk and shear modulus each draw from a suite of classical ``spring and dashpot'' rheological models. To facilitate high-fidelity viscoelastic characterization of natural and engineered solids via (i) thin-sheet testing and (ii) applications of the error-in-constitutive-relation approach to the inversion of (kinematic) sensory data, we also examine the reduction of thermodynamic potentials describing linear viscoelasticity under the plane stress condition. Here, we find that the germane reductions of the fourth-order  elasticity and dissipation tensors (specifying the thermodynamic model) follow the ``generalized-inverse-of-the-projected-inverse'' pattern featured by the preceding developments. The primary analysis is accompanied by a set of ancillary results pertaining to the issues of causality, Voigt representation of the constitutive description, and limitations of the featured (example) thermodynamic model in the context of viscoelasticity. The study is concluded by a set of examples, illustrating the effect on the plane stress condition on the behavior of both isotropic and anisotropic viscoelastic solids.

\paragraph{Acknowledgment} This work was supported as part of the \emph{Center on Geo-processes in Mineral Carbon Storage}, an Energy Frontier Research Center funded by the U.S. Department of Energy, Office of Science, Basic Energy Sciences at the University of Minnesota under award \# DE-SC0023429. Thanks are extended to MTS Systems Corporation for providing an opportunity for M. Bonnet to visit the University of Minnesota as an MTS Visiting Professor of Geomechanics. 

\bibliography{marcbibs}

\appendix

\section{Appendix}

\subsection{Matrix formulation}\label{A:Voigt}

To cater for numerical implementation, we arrange the independent components of~$\bfsig$ and~$\bfeps$ as 6-vectors, which then allows $\bfC$ to be written as a $6\times 6$ matrix. Specifically, using Voigt notation~\cite{voigt1910lehrbuch,ting:96} we write 
\begin{equation}
\begin{aligned}
\bfsigup &= \begin{array}{cccccc}
[\, \sigma_{11} & \sigma_{22} & \sigma_{33} & \sigma_{12} & \sigma_{23} & \sigma_{31} \,]\Tsup \end{array},  \qquad & \bfepsup &= \begin{array}{cccccc}
[\, \eps_{11} & \eps_{22} & \eps_{33} & 2\hh\eps_{12} & 2\hh\eps_{23} & 2\hh\eps_{31} \,]\Tsup \end{array}, \\
\bfsigup\psup  &= \begin{array}{cccccc}
[\, \sigma_{11} & \sigma_{22} & 0 & \sigma_{12} & 0 & 0 \,]\Tsup \end{array},  \qquad &\bfepsup\psup &= \begin{array}{cccccc}
[\, \eps_{11} & \eps_{22} & 0 & 2\hh\eps_{12} & 0 & 0 \,]\Tsup \end{array}, \\
\bfsigup\rsup &= \begin{array}{cccccc}
[\, 0 & 0& \sigma_{33} & 0 & \sigma_{23} & \sigma_{31} \,]\Tsup
\end{array} ,  \qquad &\bfepsup\rsup &= \begin{array}{cccccc}
[\, 0 & 0 & \eps_{33} & 0 & 2\hh\eps_{23} & 2\hh\eps_{31} \,]\Tsup \end{array} . \\
\end{aligned}
\end{equation}
As a result, the linear elastic constitutive relationship takes the matrix form $\bfsigup=\bbC\hh\bfepsup$, with $\bbC$ being given in terms of the entries of the elasticity tensor $\bfC$ by
\begin{equation}
\bbC =   \begin{bmatrix}
C_{1111}& C_{1122}& C_{1133} & C_{1112} & C_{1123} & C_{1131}\\
C_{2211}& C_{2222}& C_{2233} & C_{2212} & C_{2223} & C_{2231}\\
\vdots&\vdots&\vdots &\vdots &\vdots &\vdots \\
C_{3111}& C_{3122}& C_{3133} & C_{3112} & C_{3123} & C_{3131}
\end{bmatrix} = \bbC\Tsup.
\end{equation}
Likewise, in the matrix form the projections $\bfP$ and $\bfR$ introduced in~\eqref{newten1} read 
\begin{equation}
  \bbP = \text{diag}\! \begin{array}{cccccc} [\, 1 & 1 & 0 & 1 & 0 & 0 \,], \end{array} \qquad
  \bbR = \text{diag}\! \begin{array}{cccccc} [\, 0 & 0 & 1 & 0 & 1 & 1 \,]. \end{array}
\end{equation}
On introducing the short-hand notation
\begin{equation}
\bbA = \begin{bmatrix}
A_{11} & A_{12} & 0 & A_{13} & 0 & 0 \\
A_{21} & A_{22} & 0 & A_{23} & 0 & 0 \\
0 & 0 & 0 & 0 & 0 & 0 \\
A_{31} & A_{32} & 0 & A_{33} & 0 & 0 \\
0 & 0 & 0 & 0 & 0 & 0 \\
0 & 0 & 0 & 0 & 0 & 0
\end{bmatrix}, \quad
\bbB = \begin{bmatrix}
A_{11} & A_{12} & A_{13} \\
A_{21} & A_{22} & A_{23} \\
A_{31} & A_{32} & A_{33}
\end{bmatrix}
\qquad
\Longrightarrow \qquad
\lfloor\bbA\rfloor := \bbB, \quad \lceil\bbB\rceil := \bbA 
\end{equation}
and recalling the expressions~\eqref{strain1} and~\eqref{stress4} from Section~\ref{pspsx}, one obtains $\bbC^\eps$ (matrix version of the \emph{plane strain} elasticity tensor~$\bfC^\eps$) and $\bbC^\sigma$ (matrix version of the \emph{plane stress} elasticity tensor~$\bfC^\sigma$) respectively as
\begin{equation}
\bbC^\eps = \bbP \, \bbC \, \bbP =
\left\lceil\begin{bmatrix}
C_{1111} & C_{1122} & C_{1112}  \\
C_{2211} & C_{2222} & C_{2212} \\
C_{1211} & C_{1222} & C_{1212}
\end{bmatrix}\right\rceil
\end{equation}
and
\begin{equation}
\bbC^\sigma = (\bbP \, \bbC^{-1} \bbP)^\dagger = \lceil\,\lfloor\bbP \, \bbC^{-1} \bbP\rfloor^{-1}\,\rceil =
\left\lceil\rule{0em}{8mm}\right.\hspace*{-3pt}
\begin{bmatrix}
S_{1111} & S_{1122} & S_{1112} \\
S_{2211} & S_{2222} & S_{2212} \\
S_{1211} & S_{1222} & S_{1212}
\end{bmatrix}^{-1}
\left\rceil\rule{0em}{8mm}\right.
\end{equation}
according to Remark~\ref{MP}, where
\[
\bfC^{-1} := \bfS = S_{ijkl} \, \bfe_i\otimes\bfe_j\otimes\bfe_k\otimes\bfe_l
\] 
is the germane compliance tensor satisfying $\bfeps = \bfS\dip\bfsig$.

\paragraph{Reduced constitutive relationships}

On denoting by~$\lfloor\bfsigup\psup\rfloor$ and~$\lfloor\bfepsup\psup\rfloor$ the ``pruned'' versions of~$\bfsigup\psup$ and~$\bfepsup\psup$ that omit trivial entries, we obtain the 2D constitutive relationships
\begin{equation}
\begin{aligned}
\text{Plane strain}: \quad & \lfloor\bfsigup\psup\rfloor = \begin{bmatrix}
C_{1111} & C_{1122} & C_{1112}  \\
C_{2211} & C_{2222} & C_{2212} \\
C_{1211} & C_{1222} & C_{1212}
\end{bmatrix} \lfloor\bfepsup\psup\rfloor,
\end{aligned}\end{equation}
and
\begin{equation}
\begin{aligned}
\text{Plane stress}: \quad & \lfloor\bfsigup\psup\rfloor = \begin{bmatrix}
S_{1111} & S_{1122} & S_{1112} \\
S_{2211} & S_{2222} & S_{2212} \\
S_{1211} & S_{1222} & S_{1212}
\end{bmatrix}^{-1}  \lfloor\bfepsup\psup\rfloor.
\end{aligned}\end{equation}

\subsubsection*{Isotropic elasticity}
In the isotropic case, we have
\begin{equation}
\lfloor\bbC^\eps\rfloor = \lfloor \bbP \, \bbC \, \bbP \rfloor  = \frac{E}{1-\nu-2\nu^2}
\begin{bmatrix}
  1-\nu & \nu & 0 \\ \nu & 1-\nu & 0 \\ 0 & 0 & 1-2\nu
\end{bmatrix}
\end{equation}
for the plane strain problem, and
\begin{equation} \label{ps-matrix1}
\lfloor\bbP \, \bbC^{-1} \bbP\rfloor \,=\, \frac{1}{E}
\begin{bmatrix}
 1 & -\nu & 0 \\
-\nu & 1 & 0 \\
0 &  & 1\!+\!\nu
\end{bmatrix} \qquad\Longrightarrow\qquad
\lfloor\bbC^\sigma\rfloor \,=\, \frac{E}{1-\nu^2}
\begin{bmatrix}
 1 & \nu & 0 \\ \nu & 1 & 0 \\ 0 & 0 & 1-\nu
\end{bmatrix}
\end{equation}
for the plane stress configuration, where~$E$ and~$\nu$ are Young's modulus and Poisson's ratio. On formally rewriting~$\lfloor\bbC^\sigma\rfloor$ in the form of~$\lfloor\bbC^\eps\rfloor$, we recover the well-known expressions for the \emph{effective} plane-stress Young's modulus~$E'$ and Poisson's ratio~$\nu'$ \cite{Kolsky,Malvern} in that
\begin{equation} \label{epvp}
\lfloor\bbC^\sigma\rfloor \,=\, \frac{E'}{1-\nu'-2\nu'^2}
\begin{bmatrix}
  1-\nu' & \nu' & 0 \\ \nu' & 1-\nu' & 0 \\ 0 & 0 & 1-2\nu'
\end{bmatrix}
\qquad \text{where} \quad
\left\{ \begin{aligned}
E'  &= E\hh \frac{1\nes+\nes2\nu}{(1\nes+\nes\nu)^2} \\
\nu' &= \frac{\nu}{1\nes+\nes\nu}
\end{aligned} \right.~~.
\end{equation}

\begin{remark}
From~\eqref{ps-matrix1}, we obtain the conditions for positive definiteness of the plane-stress isotropic elasticity tensor $\lfloor\bbC^\sigma\rfloor$ as
\begin{equation}
E>0, ~~ -1<\nu< 1.
\end{equation}
One may further observe that (i) the transition from ``regular'' to auxetic material under the plane stress condition occurs at $\nu\!=\!\nu'\!=\!0$, and (ii) (macroscopically) isotropic bulk metamaterials  with $\tfrac{1}{2}\!<\!\nu\!<\!1$ -- and so negative bulk modulus -- give rise to a positive-definite elasticity tensor $\lfloor\bbC^\sigma\rfloor$.
\end{remark}

\subsection{Kramers-Kronig relationship}\label{A4}

Applying the Fourier convolution theorem given by the second of~\eqref{FCT} to definition~\eqref{H:def} of the Hilbert transform and invoking the Fourier pair~(c) in Table~\ref{tab:fourier}, we obtain $\Fcal[f\,\text{sgn}(t)](\omega)=\ii\Hcal[\hat{f}](\omega)$ and consequently 
\[
\ii\hh\Fcal[f - f\,\text{sgn}(t)](\omega) = \ii\hh\hat{f}(\omega) + \Hcal[\hat{f}](\omega). 
\]
Every causal function verifies $f - f\,\text{sgn}(t)\!=\!0$, while the Fourier transform is an $\Scal'\to\Scal'$ isomorphism. As a result, temporal function $f\!\in\!\Scal'$ is causal if and only~\cite{schwartz:62} its Fourier transform $\hat{f}$ satisfies the Kramers-Kronig relations, namely 
\begin{equation} \label{a:causality}
f\ \text{causal} \ \iff \ \Fcal[f - f\,\text{sgn}(t)] = 0  \ \iff \
\hat{f} = \ii\hh\Hcal[\hat{f}] \ \iff \
  \left\{ \begin{aligned} \Re[\hat{f}] &= -\Hcal[\Im[\hat{f}]] \\
    \Im[\hat{f}] &= \Hcal[\Re[\hat{f}]] \end{aligned} \right., 
\end{equation}
where $\Hcal[\hat{f}]$ is given by~\eqref{H:def}, understood in the Cauchy principal value sense.\enlargethispage{1ex}

\subsection{Application of the Kramers-Kronig relationship toward computing~$\alpha$ and~$\beta$ in~\eqref{Jhat:Dsing}} \label{alfabeta}

In the case (ii) of Section~\ref{relcreep}, by~\eqref{Jhat:Dsing} the Fourier image of $J(t)$ has the form
\begin{equation}
\hat{J}(\omega)
 = \ii \hh\omega^{-2}\hh f(\omega)  + \alpha\hh \delta(\omega) + \beta\hh \delta'(\omega), \label{aux11}
\end{equation}
where $f(\omega) := \Pcal(-\ii\omega)/\Rcal\cal(-\ii\omega)$ is a smooth rational function. We then may write $f(\omega)=f(0)+\omega f'(0)+\omega^2 g(\omega)$, where $g(\omega)$ is also a smooth rational function whose poles are those of $f$ and are so known (see Section~\ref{models:EDO}) to be purely imaginary. Inserting the above expansion of $f$ into~\eqref{aux11}, the Kramers-Kronig relations~\eqref{visco10cm} to be satisfied by $\hat{J}(\omega)$ reads
\begin{equation} \label{aux12}
  \ii\hh\omega^{-2}\hh f(0) + \ii\hh\omega^{-1}\hh f'(0) + \alpha\hh \delta(\omega) + \beta\hh \delta'(\omega)
 \,=\, \ii\hh\Hcal\big[ \ii\hh \omega^{-2}\hh f(0) + \ii\hh\omega^{-1}\hh f'(0) + \alpha\hh \delta(\omega) + \beta\hh \delta'(\omega) \big], 
\end{equation}
having used the fact that $g(\omega)$ itself verifies~\eqref{visco10cm} by its aforementioned properties and the proof of Proposition~\ref{KK:prop}. We then recall the known Hilbert distributional transform pairs~\cite[Vol.~2, Appendix~1, Table 1.10]{king:vol12}
\begin{equation} \label{aux12x}
  \Hcal[\delta(\omega)] = \frac{1}{\pi\omega}, \qquad
  \Hcal[\delta'(\omega)] = -\frac{1}{\pi\omega^2}, \qquad
  \Hcal[\omega^{-1}] = -\pi\delta(\omega), \qquad
  \Hcal[\omega^{-2}] = \pi\delta'(\omega),
\end{equation}
using the short-hand notations $\text{PV}\omega^{-1}\mapsto\omega^{-1}$ and $\text{FP}\omega^{-2}\mapsto\omega^{-2}$. By virtue of ~\eqref{aux12} and~\eqref{aux12x}, we obtain the distributional equality
\begin{equation}  \label{aux13x}
\big(\alpha - \pi f'(0)\big)\,
\Big(\delta(\omega)-\frac{\ii}{\pi\omega} \Big) + 
\big(\beta + \pi f(0)\big)\, 
\Big(\delta'(\omega) +\frac{\ii}{\pi\omega^2} \Big) = 0.
\end{equation}
The last causality requirement is satisfied by letting $\alpha=\pi f'(0)$ and $\beta=-\pi f(0)$, which completes the proof of~\eqref{Jhat:Dsing}.

\subsection{Fourier transform pairs} \label{A:FTP} 

We consider the Fourier transform pair
\begin{equation} \label{four1}
\hat{f}(\omega) = \mathcal{F}[f(t)](\omega) = \int_{\mathbb{R}} f(t) \hh e^{\ii\omega t} \dt, \qquad\quad
f(t) = \mathcal{F}^{-1}[\hat{f}(\omega)](t) = \frac{1}{2\pi} \int_{\mathbb{R}} \hat{f}(\omega) \hh e^{-\ii\omega t} \dom.
\end{equation}
Since $\mathcal{F}^{-1}[\delta(\omega)](t)=(2\pi)^{-1}$ and so $\mathcal{F}[1](\omega)=2\pi\hh \delta(\omega)$, by differentiating we obtain
\begin{equation} \label{four2}
\frac{\partial}{\partial \omega}\Big[
2\pi\hh \delta(\omega) = \int_{\mathbb{R}}  e^{\ii\omega t} \dt \Big] \quad \Rightarrow \quad
2\pi\hh \delta'(\omega) =\int_{\mathbb{R}}  \ii t \hh e^{\ii\omega t} \dt,
\end{equation}
where $\delta'(\omega) = \dd \delta(\omega)/\dom$ which carries he property that
\begin{equation} \label{deltaprime}
\int_{\mathbb{R}} \delta'(\omega) \hh f(\omega) \dom = - \int_{\mathbb{R}} \delta(\omega) \hh f'(\omega) \dom = -f'(0).
\end{equation}
for any $f\!\in\!C^1(\mathbb{R})$ that vanishes at~$\pm\infty$. As a result, we have
\begin{equation} \label{four3}
\mathcal{F}[t](\omega) = -2\pi\ii \hh \delta'(\omega),  \qquad
\mathcal{F}^{-1}[\delta'(\omega)](t) = \frac{\ii\hh t}{2\pi}.
\end{equation}
From~\cite{gra}, on the other hand, we recall two additional pairs
\begin{equation} \label{four4}
\mathcal{F}[\sign(t)](\omega) = \frac{2\ii}{\omega}, \qquad
\mathcal{F}^{-1}[\omega^{-1}](t) = - \frac{\ii}{2} \sign(t)
\end{equation}
(due to $\mathcal{F}[t^{-1}](\omega)=\ii\pi\hh \sign(\omega)$) and
\begin{equation} \label{four5}
\mathcal{F}[t\, \sign(t)](\omega) = -\frac{2}{\omega^2}, \qquad
\mathcal{F}^{-1}[\omega^{-2}](t) = - \frac{1}{2} t\, \sign(t).
\end{equation} 

\end{document}